\documentclass[a4paper,UKenglish,cleveref, autoref, thm-restate]{lipics-v2021}

\pdfoutput=1 
\hideLIPIcs  

\nolinenumbers

\title{On the algorithmic structure of Dialectica realisers} 

\author{Davide {Barbarossa}}{Department of Computer Science, University of Bath, UK 
\and \url{https://davidebarbarossa12.github.io/} }{db2437@bath.ac.uk}{https://orcid.org/0000-0003-4608-8282}{}

\author{Thomas {Powell}}{Department of Computer Science, University of Bath, UK 
\and \url{https://t-powell.github.io/} }{trjp20@bath.ac.uk}{https://orcid.org/0000-0002-2541-4678}{}

\authorrunning{D. {Barbarossa} and T. {Powell}} 

\Copyright{Davide Barbarossa and Thomas Powell} 


\relatedversiondetails[]{Full Version}{https://arxiv.org/abs/2501.16208} 

\funding{This work was funded by the Engineering and Physical Sciences Research Council, grant number EP/W035847/1. 
}

\acknowledgements{The authors thank Ulrich Berger (who first suggested looking at the frame rule) and Marie Kerjean (who among other things helped clarify several points from \cite{kerjean:pedrot:24:delta}).}

\usepackage{cite}
\usepackage{amsmath,amssymb,amsfonts,amsthm}
\usepackage{graphicx}
\usepackage{textcomp}
\usepackage{xcolor}
\def\BibTeX{{\rm B\kern-.05em{\sc i\kern-.025em b}\kern-.08em
    T\kern-.1667em\lower.7ex\hbox{E}\kern-.125emX}}

\usepackage{enumerate}
\usepackage{url}
\usepackage{hyperref}
\usepackage{virginialake}
\usepackage{bm}

\usepackage{algorithm}
\usepackage{algpseudocode}

\newcommand{\eha}{\mathrm{E}\mbox{-}\mathrm{HA}^\omega}
\newcommand{\weha}{\mathrm{WE}\mbox{-}\mathrm{HA}^\omega}

\newcommand{\ac}{\mathrm{AC}}
\newcommand{\ip}{\mathrm{IP}^\omega_\forall}
\newcommand{\br}{\mathrm{BR}}
\newcommand{\markov}{\mathrm{M}^\omega}

\newcommand{\hr}[3]{{#1}\,{#2}\,{#3}}
\newcommand{\dpair}[2]{\langle{#1}\, | \, {#2}\rangle}
\newcommand{\hrd}[4]{\hr{#1}{\dpair{#2}{#3}}{#4}}

\newcommand{\ite}[3]{\mathtt{if}({#1},{#2},{#3})}

\newcommand{\ax}{\mathrm{Ax}}
\newcommand{\rul}{\mathrm{Rule}}

\newcommand{\sys}{\mathcal{H}}
\newcommand{\set}[1]{\{#1\}}
\newcommand{\hoare}[3]{\set{#1}\, {#2} \,\set{#3}} 

\newcommand{\tup}[1]{\bm{#1}}

\newcommand{\rec}{\mathtt{rec}}
\newcommand{\whilerec}[3]{\mathtt{whilerec}_{#1,#2,#3}}
\newcommand{\whilerecax}[1]{\mathtt{whilerec}_{#1}}
\renewcommand{\whiledo}[3]{\mathtt{while}^{+}_{#1}\, {#2}\, \mathtt{do}\, {#3}}
\newcommand{\whiledoback}[4]{\mathtt{while}^{-}_{#1}\, {#2}\, \mathtt{do}\, {(#3,#4)}}

\newcommand{\loopi}{\mathrm{LOOP}}
\newcommand{\loopd}{\mathrm{LOOP}_D}
\newcommand{\comf}{\mathtt{Comm}^{+}}
\newcommand{\comb}{\mathtt{Comm}^{-}}
\newcommand{\comm}{\mathtt{Comm}}
\newcommand{\rel}{\mathtt{Rel}}
\newcommand{\expr}{\mathtt{Exp}}
\newcommand{\loopite}[3]{\mathtt{if}\, {#1}\, \mathtt{then} \, {#2} \, \mathtt{else}\, {#3}}
\newcommand{\loopwhiledo}[3]{\mathtt{while}_{#1}\, {#2}\, \mathtt{do} \, {#3}}
\newcommand{\seq}[2]{{#1}\, ; \, {#2}}
\newcommand{\loopskip}{\mathtt{skip}}
\newcommand{\totalhoare}[3]{[{#1}]\, {#2}\, [{#3}]}

\newcommand{\backward}[6]{{#1},{#2},{#3}\,{\Downarrow^-}\, {#4},{#5},{#6}}
\newcommand{\forwards}[5]{{#1},{#2}\, {\Downarrow^+}\, {#3},{#4},{#5}}

\newcommand{\diff}[2]{D_{#2}({#1})}
\newcommand{\rdiff}[2]{D^\ast_{#2}({#1})}
\newcommand{\x}{\mathtt{x}}

\newcommand{\nat}{\mathsf{nat}}
\newcommand{\zero}{\mathsf{0}}
\newcommand{\suc}{\mathsf{suc}}

\newcommand{\dhl}{\mathrm{DHL}}

\newcommand{\ind}[1]{\mathrm{I}_{#1}}

\newcommand{\dt}[3]{|{#1}|^{#2}_{#3}}

\EventEditors{Stefano Guerrini and Barbara K\"{o}nig}
\EventNoEds{2}
\EventLongTitle{34th EACSL Annual Conference on Computer Science Logic (CSL 2026)}
\EventShortTitle{CSL 2026}
\EventAcronym{CSL}
\EventYear{2026}
\EventDate{February 23--28, 2026}
\EventLocation{Paris, France}
\EventLogo{}
\SeriesVolume{363}
\ArticleNo{48}

\begin{document}

\maketitle

\begin{abstract}
G\"odel's Dialectica interpretation is a fundamental tool for the extraction of computational content from proofs, and plays a central role in today's proof mining program. In the past decades, it has also been studied from the perspective of programming languages, and our contribution is in that direction. Specifically, we present Dialectica as a collection of rules in the style of Hoare logic, where Dialectica is now viewed as a language for specifying procedural programs that come with a forward and backward direction. This viewpoint captures the interesting dynamics of realisers extracted by the Dialectica interpretation, and we illustrate this by defining a generalised backpropagation semantics for a fragment of this language. We envisage this work as providing a base for several future developments, both theoretical and practical, which we outline at the end. 
\end{abstract}

\section{Introduction}
\label{sec:intro}

G\"odel's Dialectica interpretation \cite{goedel:58:dialectica} (which we will often refer to as just ``Dialectica'') is one of the most important methods for extracting computational content from proofs. Interest in this technique has grown rapidly in recent years, partially due to increased activity in two distinct strands of research.

The first is the proof mining program (an overview of which can be found in the monograph \cite{kohlenbach:08:book} or the survey papers \cite{kohlenbach:17:recent,kohlenbach:19:nonlinear:icm}), in which methods from proof theory are applied to obtain both quantitative and qualitative improvements of theorems in mainstream mathematics through the analysis of their proofs. Here, the Dialectica interpretation is used in particular to formulate so-called logical metatheorems, which involve extensions of Dialectica to sophisticated proof systems designed to explain the extractability of quantitative information in a specific setting (see e.g.~the classic works \cite{kohlenbach:05:metatheorems,kohlenbach-gerhardy:08:metatheorems}, or the most recent metatheorem in probability theory \cite{neri-pischke:25:formal}).
    
    The second, somewhat orthogonal strand, is represented by a variety of works that approach Dialectica from a programming or categorical point of view, where the focus is on its \emph{structural} properties as a logical and program transformation, rather than on its concrete applications to mathematics.
    A crucial ingredient here is de Paiva's Dialectica categories and the discovery of the links between Dialectica and linear logic \cite{depaiva:91:thesis} (recent developments of the categorical viewpoint can be found in e.g.\ \cite{trotta:etal:pp:fibrations}).
   Most relevant to this paper is the recent reformulation of Dialectica as a genuine higher-order functional program transformation \cite{pedrot:14:functional}, which also led to the discovery of its links with automatic differentiation and the differential $\lambda$-calculus \cite{kerjean:pedrot:24:delta}.

This paper is a study of the Dialectica interpretation that can be related to both of these perspectives: We consider Dialectica in its traditional form as used in proof mining, but with a focus on the \emph{algorithmic structure} of the terms it extracts (referred to here as ``Dialectica realisers''). More specifically, we give an alternative presentation of the standard Dialectica interpretation, through a set of rules that that manipulate programs rather than logical formulas. The rules are set up in the style of Hoare logic \cite{hoare:69:logic}, and act on what we call \emph{Dialectica triples}, which are for now just realizing terms for implications between formulas (the characterising feature of Dialectica), but which we later connect with Hoare triples in the usual sense. Our rules formulate the standard soundness theorem for Dialectica by focusing on properties of extracted realizers, and in this way we seek to expose some of the elegant patterns and symmetries that govern programs extracted by the Dialectica.

We build on this perspective in two ways. First, we introduce a while loop for Dialectica into our term language, and show that it interprets a corresponding rule. We argue that this can be used to neatly describe iterative programs that arise from nonconstructive proofs in mathematics, where it is connected to the idea of interpreting classical proofs via \emph{learning} or \emph{backtracking}, a notion that dates back to Hilbert's substitution method \cite{ackermann:24:substitution}.

The fact that Dialectica admits a Hoare Logic presentation or, dually, that Hoare logic admits an extension as a logic for Dialectica realisers, raises the question of whether we can understand them from a procedural perspective. We take a first step in this direction in Section \ref{sec:imperative},
where we consider a restricted set of proof rules whose Dialectica realisers can be regarded as stateful in a suitable sense (though we hesitate to characterise them as truly imperative due to the lack of variable assignment for now, discussed further in Section~\ref{sec:conc}). We describe an operational semantics which demonstrates that our programs perform a generalised backpropagation algorithm, comprising a conventional forward part together with a backward pass that computes the reverse witness.

Our aim in all of this is to provide a perspective that can be properly developed in future work, and in that spirit, we conclude by discussing several concrete ideas, covering both theory and applications.

\subparagraph*{Our choice of formal system.}
We will generally privilege a standard presentation of the Dialectica, over more program theoretic ones as found in e.g.\ \cite{pedrot:14:functional}. In particular, we work in a higher-order version of Heyting arithmetic 
rather than a ``modern'' type theory. Extensions of the former have proven to be more than capable of capturing the application of Dialectica to a great variety of subsystems of ordinary mathematics: An instructive recent example is probability theory, which a-priori requires strong set-theoretic principles to carry out even the most basic tasks, but which is nevertheless amenable to program extraction via the Dialectica as set out in \cite{neri-pischke:25:formal} and as exemplified by the numerous case studies that have appeared in the last years. It is open whether alternative systems based on e.g.\ dependent type theory are better suited to the very precise task of capturing applications of Dialectica to \emph{tame} proofs in ordinary mathematics (see \cite{kohlenbach:20:tameness} for a discussion of ``proof theoretic tameness'' in this context), or can be as easily combined with logical relations such as majorizability \cite{howard:73:majorizability} that are essential for capturing uniformities in extracted programs. 
Finally, even though as shown in  \cite{pedrot:14:functional} these traditional formulations usually lack a crucial proof/program theoretic property, namely they break $\beta$-equivalence, here we embrace this fact at it gives rise to a natural \emph{if-then-else} case distinction construct which fits well with our procedural perspective. 

\section{Preliminaries}
\label{sec:prelim}


\paragraph*{The base system $\weha$}
\label{sec:prelim:base}

Essentially all applications of Dialectica to concrete proofs in mathematics 
can be described formally in terms of some theory based on arithmetic in all finite types, and accordingly we take as our base theory the weakly extensional version $\weha$ of higher-order Heyting arithmetic, the higher-order theory of intuitionistic arithmetic with all finite types, whose underlying programming language of terms corresponds to System T.
This can be simultaneously used as a system for formalising mathematics, or as a logic for reasoning about higher-order programs, and we take both perspectives in this paper. 
Full details can be found in \cite[Chapter 3]{kohlenbach:08:book}, and the very brief overview here is merely used as an opportunity to fix notational conventions. The types of $\weha$ are the simple types 
\[
X,Y::=\nat\, | \, X\to Y
\]
where $\nat$ represents a base type of natural numbers. As in conventional in proof mining, we work with sequences of types in the metalanguage rather than explicitly introducing product types, which in particular allows us to employ more liberal abbreviations in the logical system. We use boldface $\tup{X}=X_1,\ldots, X_n$ to denote sequences of types, and from now on, when we say `type' we usually refer to a sequence. Similarly, we use boldface $\tup{a}=a_1,\ldots,a_n$ to denote sequences of terms, writing $\tup{a}:\tup{X}$ or $\tup{a}^{\tup{X}}$ to denote $a_1:X_1,\ldots,a_n:X_n$. 

The terms of $\weha$ are those of the simply typed $\lambda$-calculus plus constants $\zero:\nat$ and $\suc:\nat \to \nat$ for zero and successor (we use $x+1$ for $\suc\, x$) and, for each sequence $\tup{X}$ of types, a  constant $\tup{\rec}_{\tup{X}}:(\nat\to\tup{X}\to\tup{X})\to \tup{X}\to \nat\to\tup{X}$ for primitive recursion.
We include an explicit cases constructor $\ite{b^\nat}{\tup{s}}{\tup{t}}$ for all types, even though this is definable from the recursor. We make free use of a number of standard abbreviations around sequences of types and terms. If $\tup{X}=X_1,\ldots,X_n$ and $\tup{Y}=Y_1,\ldots,Y_m$ then $\tup{X}\to \tup{Y}$ denotes the sequence $(X_1\to\ldots\to X_n\to Y_j)^m_{j=1}$, Similarly, if $\tup{a}:\tup{X}\to \tup{Y}$ and $\tup{b}:\tup{X}$, then by $\tup{a}\tup{b}:\tup{Y}$ we mean $(a_jb_1\ldots b_n)^m_{j=1}$, and if $\tup{c}:\tup{Y}$ are terms and $\tup{x}:\tup{X}$ are variables, then by $\lambda\tup{x}.\tup{c}:\tup{X}\to \tup{Y}$ we mean$(\lambda x_1,\ldots,x_n.c_j)^m_{j=1}$. We write $\tup{b},\tup{c}:\tup{X},\tup{Y}$ for the concatenation of two sequences. 

Formulas of $\weha$ are built from atomic formulas of the form $t=_\nat s$ for $t,s:\nat$, 
the usual logical connectives $\vee,\wedge,\to,\top,\bot$ (where we write $\neg A:=A\to \bot$), and quantifiers $\forall^X,\exists^X$ for each simple
type $X$ (which we usually omit). We also write $\exists \tup{x}\, A(\tup{x})$ for $\exists^{X_1} x_1,\ldots,\exists^{X_n} x_n\, A(x_1,\ldots,x_n)$ and similarly for $\forall$. Equality $=_X$ at higher types is defined in terms of $=_\nat$, where for $X=X_1 \to \ldots \to X_n\to \nat$ we set 
\[
t=_X s:=\forall x_1,\ldots,x_n(tx_1\ldots x_n=_\nat sx_1\ldots x_n).
\]
Equality for sequences is defined in the obvious way. The axioms and rules of $\weha$ are those of usual intuitionistic logic extended to all simple types, 
along with equality axioms for $=_\nat$, induction, and axioms for the arithmetical constants and terms of System T.
For instance, the conditional satisfies
\[
b=_\nat 0\to \ite{b}{\tup{s}}{\tup{t}}=\tup{s} \ \ \ \mbox{and} \ \ \ b\neq_{\nat} 0\to \ite{b}{\tup{s}}{\tup{t}}=\tup{t}
\]
and the axioms for the recursors are:
\[
\rec\,\tup{z}\tup{y} 0=\tup{y} \ \ \ \mbox{and} \ \ \ \rec\,\tup{z}\tup{y} (x+1)=\tup{z}x(\rec\,\tup{z}\tup{y}  x).
\]
For equality, we consider the weak extensionality rule\footnote{We stress that extensionality issues are not a concern for what we do: Soundness theorems for Dialectica are formulated using $\weha$ for the simple reason that the extensionality \emph{axiom} is not admissible. However, because our focus in on a descriptive system for realizers, we could equally well work in the fully extensional version of Heyting arithmetic $\eha$.}
\[
\frac{A_0\to a=_X b}{A_0\to B[a/x]\to B[b/x]}
\]
where $B$ is an arbitrary formula and $A_0$ is quantifier free. For full details of $\weha$ the reader is directed to \cite{kohlenbach:08:book}, though exact details at this level are unimportant for what follows.

\paragraph*{The Dialectica interpretation}
\label{sec:prelim:dialectica}

We give a brief overview of the Dialectica interpretation of \emph{formulas} (as opposed to \emph{proofs}). This is completely standard, and full details can be found in \cite[Chapter 8]{kohlenbach:08:book}. In simple terms, the Dialectica interpretation is a mapping $P\mapsto \exists\tup{x}\forall\tup{y}\dt{P}{\tup{x}}{\tup{y}}$ that transforms formulas from some source logic (originally Heyting arithmetic) to some higher-order functional language (originally System~T).
The transformed formula $\dt{P}{\tup{x}}{\tup{y}}$ contains two additional tuples of free variables, playing the role of witnesses $\tup{x}$ and counterwitnesses $\tup{y}$ for $P$, while the formula $\dt{P}{\tup{x}}{\tup{y}}$ can be viewed as an \emph{orthogonality relation} similarly to classical realisability.
The aim is to establish that whenever $P$ is provable in the source system, there exists terms $\tup{a}$ such that $\forall\tup{y}\dt{P}{\tup{a}}{\tup{y}}$ is provable in the target system. For a term $t:\nat$, we introduce a new abbreviation $\vee_t$ for the formula
\[
A\vee_t B:=(t=_\nat 0\to A)\wedge (t\neq_\nat 0\to B)
\]
Using the fact that any quantifier-free formula $\phi$ of $\weha$ can be represented by a ``characteristic term'' 
$\chi_\phi:\nat$ with the same free variables satisfying $\chi_\phi=0\leftrightarrow \phi$, we extend this connective to such formulas by defining $A\vee_\phi B:=A\vee_{\chi_\phi} B$.

\begin{definition}[Dialectica interpretation \cite{goedel:58:dialectica}, see also \cite{kohlenbach:08:book}]
\label{def:dial}
For a formula $P$ of $\weha$, we define its Dialectica interpretation to be the formula $\exists \tup{x}\forall\tup{y}\,\dt{P}{\tup{x}}{\tup{y}}$, whose free variables are the same as $P$, and where $\dt{P}{\tup{x}}{\tup{y}}$ is a quantifier-free formula, defined inductively along with the types of $\tup{x}$ and $\tup{y}$ as follows:\smallskip
\begin{itemize}

	\item $\dt{P}{}{}:=P$ if $P$ is atomic
	
	\item $\dt{P\wedge Q}{\tup{x},\tup{u}}{\tup{y},\tup{v}}:=\dt{P}{\tup{x}}{\tup{y}}\wedge \dt{Q}{\tup{u}}{\tup{v}}$
	
	\item $\dt{P\vee Q}{b,\tup{x},\tup{u}}{\tup{y},\tup{v}}:=\dt{P}{\tup{x}}{\tup{y}}\vee_b \dt{Q}{\tup{u}}{\tup{v}}$
	
	\item $\dt{P\to Q}{\tup{f},\tup{F}}{\tup{x},\tup{v}}:=\dt{P}{\tup{x}}{\tup{F}\tup{x}\tup{v}}\to \dt{Q}{\tup{f}\tup{x}}{\tup{v}}$
	
	\item $\dt{\exists x\, P(x)}{x,\tup{u}}{\tup{y}}:=\dt{P(x)}{\tup{u}}{\tup{y}}$
	
	\item $\dt{\forall x\, P(x)}{\tup{f}}{x,\tup{y}}:=\dt{P(x)}{\tup {f}x}{\tup{y}}$\smallskip

\end{itemize}
Note that if $\phi$ is quantifier-free then we can assume that it doesn't contain $\vee$ (as it can be rewritten as $\chi_\phi=0$) and therefore that $\dt{\phi}{}{}=\phi$.
\end{definition}
The characteristic feature of Dialectica, which sets it apart from realisability-like interpretations, is its treatment of implication, which comprises a forward-component $\tup{f}$ as in realisability along with a backward component $\tup{F}$, which together compose in a ``backpropagation'' style. The soundness theorem for the Dialectica interpretation (i.e. the Dialectica interpretation of \emph{proofs}), for extensions of $\weha$, usually takes the following general form:
\begin{theorem}[Generic soundness theorem]
\label{res:soundness:dial}
Suppose that $\Delta$ is a set of axioms whose Dialectica interpretation is witnessed by terms of $\weha_\Delta$ (provably in this system) for some suitable extension of $\weha$, and let $\mathcal{U}$ be a set of purely universal formulas. Then whenever
\[
\weha+\Delta+\mathcal{U}\vdash P
\] 
we can extract terms $\tup{a}$, whose free variables are the same as those of $P$, such that
\[
\weha_\Delta+\mathcal{U}\vdash \forall\tup{y}\dt{P}{\tup{a}}{\tup{y}}.
\]
\end{theorem}
In particular, we can set $\Delta=\mathcal{U}=\emptyset$ to obtain the Dialectica interpretation of $\weha$. A common formulation (cf. \cite[Theorem 8.6]{kohlenbach:08:book}) has $\Delta$ consisting of the axiom of choice ($\ac$), the independence of premise scheme for universal formulas ($\ip$), and Markov's principle ($\markov$) with $\weha_\Delta=\weha$. A further extension is to add the double negation shift to $\Delta$, and then $\weha_\Delta=\weha+(\br)$ i.e. Heyting arithmetic extended with bar recursion in all finite types (cf. \cite[Chapter 11]{kohlenbach:08:book}). 

\section{The action of Dialectica on programs}
\label{sec:dynamics}

We now present the soundness theorem for the Dialectica interpretation from a slightly different perspective, by providing a collection of specification rules for Dialectica realizers (terminology we use for programs extracted from proofs using Dialectica), all of which are sound in our base theory $\weha$. The intention here is to construct a rich, descriptive language where the focus is on the realizing terms rather than the underlying logic, one which can be easily extended with new programming constructs, or new rules for describing extracted programs of a specific type. We give several examples of this in what follows.

Our basic approach is inspired by Hoare logic \cite{hoare:69:logic}, in the sense that our rules apply to ``triples'' that describe properties of Dialectica realizers. Let us explain this in more detail. Viewed abstractly, a basic Hoare triple $\hoare{P}{c}{Q}$ describes the effect of some command $c$ in an imperative language on an underlying state. Here $P$ (the precondition) and $Q$ (the postcondition) are assertions about the state, and the meaning of the triple is that whenever the command executes on a state satisfying $P$, the resulting state satisfies $Q$. We can represent this basic idea in predicate logic by viewing the pre- and postconditions as formulas $P(s),Q(s)$ acting on objects $s:S$ for some state type $S$, and commands as functions $c:S\to S$. The Hoare triple then becomes the statement that $c$ satisfies
\[
\forall s(P(s)\to Q(cs)).
\]
If we now imagine that $P(s)$ and $Q(s)$ are quantifier-free formulas, and the state can be encoded in a suitable way in $\weha$, the formula $P(s)\to Q(cs)$ is exactly
\[
\dt{\exists x\, P(x)\to \exists x\, Q(x)}{c}{}.
\] 
With this loose correspondence in mind, we now generalise to formulas of arbitrary logical complexity and define a ``Dialectica triple'' as
\[
\hrd{P}{\tup{a}}{\tup{\alpha}}{Q}:=\forall \tup{x},\tup{v}\, \dt{P\to Q}{\tup{a},\tup{\alpha}}{\tup{x},\tup{v}}.
\]
where now $\tup{a}$ and $\tup{\alpha}$ are sequences of terms of some type determined by the formulas $P$ and $Q$. 
In terms of actions on a hypothetical state $\tup{x}$, we visualise the pair $\dpair{\tup{a}}{\tup{\alpha}}$ as comprising two components: A forward ``command'' $\tup{a}$ satisfying
\[
\forall\tup{x}\left(\forall\tup{y}\dt{P}{\tup{x}}{\tup{y}}\to \forall\tup{v}\dt{Q}{\tup{a}\tup{x}}{\tup{v}}\right)
\]
followed by a backward command $\tup{\alpha}$, that operates on a ``dual state'' $\tup{v}$, satisfying
\[
\forall\tup{x},\tup{v}\left(\neg\dt{Q}{\tup{a}\tup{x}}{\tup{v}}\to \neg\dt{P}{\tup{x}}{\tup{\alpha}\tup{x}\tup{v}}\right)
\]
This approach differs from the use of Hoare triples for realizability in \cite{powell:24:hoare}, which is based on a monadic translation of proofs: Here Hoare triples are just the interpretation of implication, and a corresponding stateful interpretation in terms of backpropagation will be discussed in more detail in Section \ref{sec:imperative} below. 

As we will see, our use of Hoare logic \emph{in the abstract} allows us to view Dialectica realisers as `stateful' or \emph{procedural} programs. However, we are cautious in drawing too strong a connection: The fundamental power of Hoare logic and its various extensions lies in its ability to reason explicitly about the structure of the state, something that in connection to Dialectica we leave to future work. Therefore at present, we do not characterise Dialectica realisers as truly \emph{imperative}. However, we do believe that such view is possible and we think of this work as the first step towards it.

\subsection{A Dialectica Hoare logic ($\dhl$)}
\label{sec:dynamics:rules}

\begin{figure*}
\fbox{
\begin{minipage}{0.99\textwidth}
\scriptsize
$\begin{gathered}
\textbf{Propositional rules: Axioms, basic actions, conditionals, switching, composition.}
\\[2mm]
\hrd{\bot}{\tup{a}}{-}{P}
\qquad
\hrd{P}{-}{\tup{\alpha}}{\top}
\qquad
\hrd{P}{\lambda\tup{x}.\tup{x}}{\lambda\tup{x},\tup{v}.\tup{v}}{P}
\qquad
\vlinf{}{}
	{\hrd{P_\exists}{-}{-}{Q_\forall}}
	{P_\exists\to Q_\forall\in\ax}
\qquad
\vlinf{}{}
	{\hrd{P'_\exists}{-}{-}{Q'_\forall}}
	{\hrd{P_\exists}{-}{-}{Q_\forall}}
\text{\,  for $\frac{P_\exists\to Q_\forall}{P'_\exists\to Q'_\forall}\in\rul$}
\\[2mm]
\vlinf{}{p\wedge_R}
	{\hrd{P}{\tup{b},\tup{a}}{\tilde{\tup{\alpha}}}{R\wedge Q}}
	{\hrd{P}{\tup{a},\tup{b}}{\tup{\alpha}}{Q\wedge R}}
\qquad
\vlinf{}{p\wedge_L}
	{\hrd{Q\wedge P}{\tilde{\tup{a}}}{\tilde{\tup{\beta}},\tilde{\tup{\alpha}}}{R}}
	{\hrd{P\wedge Q}{\tup{a}}{\tup{\alpha},\tup{\beta}}{R}}
\qquad
\vlinf{}{p\vee_R}
	{\hrd{P}{\tup{b},\tup{a}}{\tilde{\tup{\alpha}}}{R\vee_{\bar{c}} Q}}
	{\hrd{P}{\tup{a},\tup{b}}{\tup{\alpha}}{Q\vee_c R}}
\qquad
\vlinf{}{p\vee_L}
	{\hrd{Q\vee_{\bar{c}} P}{\tilde{\tup{a}}}{\tilde{\tup{\beta}},\tilde{\tup{\alpha}}}{R}}
	{\hrd{P\vee_c Q}{\tup{a}}{\tup{\alpha},\tup{\beta}}{R}}
\\[2mm]
\vlinf{}{\vee_R}
	{\hrd{P}{\tup{a},\tup{b}}{\tup{\alpha}_\pi}{Q\vee_0 R}}
	{\hrd{P}{\tup{a}}{\tup{\alpha}}{Q}}
\qquad
\vlinf{}{\wedge_L}
	{\hrd{P\wedge R}{\tup{a}_\pi}{\tup{\alpha}_\pi,\tup{\beta}}{Q}}
	{\hrd{P}{\tup{a}}{\tup{\alpha}}{Q}}
\qquad
\vlinf{}{\wedge_R}
	{\hrd{P}{\tup{a}}{\tup{\alpha}_{p}}{Q}}
	{\hrd{P}{\tup{a},\tup{b}}{\tup{\alpha}}{Q\wedge R}}
\qquad
\vlinf{}{\vee_L}
	{\hrd{P}{\tup{a}_p}{\tup{\alpha}_p}{Q}}
	{\hrd{P\vee_0 R}{\tup{a}}{\tup{\alpha},\tup{\beta}}{Q}}
\\[2mm]
\vliinf{}{cond_L}
	{\hrd{P\vee_\phi Q}{\lambda\tup{x},\tup{y}.\ite{\phi}{\tup{a}\tup{x}}{\tup{b}\tup{y}}}{\tup{\alpha}_\pi,\tup{\beta}_\pi}{R}}
	{\hrd{P\wedge\phi}{\tup{a}}{\tup{\alpha}}{R}}
	{\hrd{Q\wedge\neg\phi}{\tup{b}}{\tup{\beta}}{R}}
\qquad
\vliinf{}{cond_R}
	{\hrd{P}{\tup{a},\tup{b}}{\lambda\tup{x},\tup{v},\tup{w}.\ite{\dt{P}{\tup{x}}{\tup{\alpha}\tup{x}\tup{v}}}{\tup{\beta}\tup{x}\tup{w}}{\tup{\alpha}\tup{x}\tup{v}}}{Q\wedge R}}
	{\hrd{P}{\tup{a}}{\tup{\alpha}}{Q}}
	{\hrd{P}{\tup{b}}{\tup{\beta}}{R}}
\\[2mm]
\vlinf{}{imp}
	{\hrd{P\wedge Q}{\tup{a}}{\tup{\alpha},\tup{b}}{R}}
	{\hrd{P}{\tup{a},\tup{b}}{\tup{\alpha}}{Q\to R}}
\qquad
\vlinf{}{exp}
	{\hrd{P}{\tup{a},\tup{\beta}}{\tup{\alpha}}{Q\to R}}
	{\hrd{P\wedge Q}{\tup{a}}{\tup{\alpha},\tup{\beta}}{R}}
\qquad
\vliinf{}{comp}
	{\hrd{P}{\tup{b}\circ\tup{a}}{\tup{\alpha}\ast_{\tup{a}}\tup{\beta}}{R}}
	{\hrd{P}{\tup{a}}{\tup{\alpha}}{Q}}
	{\hrd{Q}{\tup{b}}{\tup{\beta}}{R}}
\qquad
\\[2mm]
\textbf{Quantifier rules: Term introduction, $\lambda$-abstraction and application, epsilon terms.}
\\[2mm]
\vlinf{}{\exists_R}
	{\hrd{P}{\lambda\_.\tup{t},\tup{a}}{\tup{\alpha}}{\exists \tup{x}\, Q(\tup{x})}}
	{\hrd{P}{\tup{a}}{\tup{\alpha}}{Q(\tup{t})}}
\qquad
\vlinf{}{\forall_L}
	{\hrd{\forall \tup{x}\, P(\tup{x})}{\lambda\tup{f}.\tup{a}(\tup{f}\tup{t})}{\lambda\_.\tup{t},\lambda\tup{f}.\tup{\alpha}(\tup{f}\tup{t})}{Q}}
	{\hrd{P(\tup{t})}{\tup{a}}{\tup{\alpha}}{Q}}
\\[2mm]
\vlinf{}{\exists_L}
	{\hrd{\exists \tup{x}\, P(\tup{x})}{\lambda \tup{x}.\tup{a}}{\lambda \tup{x}.\tup{\alpha}}{Q}}
	{\hrd{P(\tup{x})}{\tup{a}}{\tup{\alpha}}{Q}}
\qquad
\vlinf{}{\forall_R}
	{\hrd{P}{\lambda\tup{y},\tup{x}.\tup{a}\tup{y}}{\lambda\tup{y},\tup{x}.\tup{\alpha}\tup{y}}{\forall\tup{x}\, Q(\tup{x})}}
	{\hrd{P}{\tup{a}}{\tup{\alpha}}{Q(\tup{x})}}
\qquad
\text{$\tup{x}$ not free in $Q$ resp. $P$ for $\exists_L$ resp. $\forall_R$}
\\[2mm]
\!\vlinf{}{s_L}
	{\hrd{P(\tup{t})}{\tup{a}\tup{t}}{\tup{\alpha}\tup{t}}{Q}}
	{\exists \tup{x}\, \hrd{P(\tup{x})}{\tup{a}}{\tup{\alpha}}{Q}}
\qquad
\!\vlinf{}{s_R}
	{\hrd{P}{\lambda\tup{y}.\tup{a}\tup{y}\tup{t}}{\lambda\tup{y},\tup{v}.\tup{\alpha}\tup{y}\tup{t}\tup{v}}{Q(\tup{t})}}
	{\hrd{P}{\tup{a}}{\tup{\alpha}}{\forall \tup{x}\, Q(\tup{x})}}
\qquad
\!\vlinf{}{\epsilon_R}
	{\hrd{P_\forall}{\tup{b}}{\tup{\alpha}}{Q(\tup{a})}}
	{\hrd{P_\forall}{\tup{a},\tup{b}}{\tup{\alpha}}{\exists \tup{x}\, Q(\tup{x})}}
\qquad
\!\!\vlinf{}{\!\epsilon_L}
	{\hrd{P_\forall(\tup{\alpha})}{-}{\tup{\beta}}{Q_{qf}}}
	{\hrd{\forall \tup{x}\, P_\forall(\tup{x})}{-}{\tup{\alpha},\tup{\beta}}{Q_{qf}}}
\\[2mm]
\textbf{Consequence, extensionality, and induction/recursion}
\\[2mm]
\vliiinf{}{cons}
	{\hrd{P'}{\tup{a}}{\tup{\alpha}}{Q'}}
	{P'\to_D P}
	{\hrd{P}{\tup{a}}{\tup{\alpha}}{Q}}
	{Q\to_D Q'}
\qquad
\vliinf{}{ext}
	{\hrd{P}{\tup{b}}{\tup{\beta}}{Q}}
	{\hrd{P}{\tup{a}}{\tup{\alpha}}{Q}}
	{\tup{a},\tup{\alpha}=\tup{b},\tup{\beta}}
\qquad
\vlinf{}{ind}
	{\hrd{P(0)}{\rec\, {\tup{a}}}{\rec^\ast\tup{a}\tup{\alpha}}{\forall x\, P(x)}}
	{\hrd{P(x)}{\tup{a}(x)}{\tup{\alpha}(x)}{P(x+1)}}
\end{gathered}$
\end{minipage}
}
\caption{Rules for Dialectica triples, definitions and abbreviations discussed in Section \ref{sec:dynamics:rules}.}
\label{fig:dynamics}
\end{figure*}

It is natural to ask what Hoare-like rules govern Dialectica realizers when viewed in this way. Fixing sets $\ax,\rul$ of axioms and rules in $\weha$, we define a judgment system $\dhl$ (for \emph{Dialectica Hoare Logic}) which derives judgments of shape $P\dpair{\tup{a}}{\tup{\alpha}}Q$, where $P,Q$ are formulas of $\weha$ and $\tup{a},\tup{\alpha}$ are (sequences of) terms of System T, according to the rules in Figure~1.
In all cases, the types of $\tup{a},\tup{\alpha}$ are left implicit, but can be inferred from the rules and formulas, as usual in Dialectica. These rules are described in detail in the points below, where we also clarify the notations and abbreviations used.

\begin{itemize}

	\item We write $P_\exists$ to denote a formula whose Dialectica interpretation is purely existential i.e. of the form $\exists\tup{x}\, \dt{P}{\tup{x}}{}$. Similarly, $P_\forall$ and $P_{qf}$ mean that the Dialectica interpretation of $P$ is purely universal resp.\ quantifier-free.
	
	\item The sets $\ax$ and $\rul$ denote sets of universal axioms and rules that we add to the system: At the very least $\ax$ would include the axioms and rules for $=_\nat$ along with those governing System T terms, while $\rul$ would include the quantifier-free extensionality rule, but we leave open the possibility that other axioms and rules could be included. 
	
	\item For $c:\nat$ we define $\bar{c}$ by $\bar{0}:=1$ and $\overline{n+1}:=0$.

	
	\item $\tilde{\tup{\alpha}}$ denotes a permutation of the arguments of $\tup{\alpha}$, whose precise definition varies but can be immediately inferred from the rule. For example, written out fully, the $\tilde{\tup{\alpha}}$ in the conclusion of $(p\wedge_R)$ should be $\tilde{\tup{\alpha}}:=\lambda\tup{x},\tup{w},\tup{v}.\tup{\alpha}\tup{x}\tup{v}\tup{w}$
    and the intuitive meaning of the triple is the formula $\forall\tup{x},\tup{w},\tup{v}\left(\dt{P}{\tup{x}}{\tilde{\tup{\alpha}}\tup{x}\tup{w}\tup{v}}\to \dt{R}{\tup{b}\tup{x}}{\tup{w}}\wedge\dt{Q}{\tup{a}\tup{x}}{\tup{v}}\right)$.
	This further illustrates our philosophy on implicit typing: The aim of our notation is to suppress bureaucratic $\lambda$-terms wherever 
    these can be directly inferred.

	\item In a similar way, $\tup{\alpha}_\pi$ denotes a projection (e.g. in $(\vee_R)$, $\tup{\alpha}_\pi$ is shorthand for $\lambda\tup{x},\tup{y},\tup{v}.\tup{\alpha}\tup{x}\tup{v}$), while $\tup{\alpha}_p$ denotes a coprojection, by which we mean the instantiation of certain arguments with canonical zero terms $\tup{0}$ of the right type (e.g. in $(\wedge_R)$, $\tup{\alpha}_p$ is shorthand for $\lambda\tup{x},\tup{v}.\tup{\alpha}\tup{x}\tup{v}\tup{0}$).
	
	\item $\lambda\_.t$ denotes the constant term $\lambda\tup{x}.\tup{t}$, where the types of $\tup{x}$, not free in $t$, are to be inferred.
	
	\item For the conditional rules, $\phi$ is always quantifier-free.
	
	\item For composition, $\circ$ denotes the usual composition of functions, while $\ast$ denotes the special backwards composition for the Dialectica interpretation, with
	\[
	\tup{\alpha}\ast_{\tup{a}}\tup{\beta}:=\lambda\tup{x},\tup{w}.\tup{\alpha}\tup{x}\left(\tup{\beta}(\tup{a}\tup{x})\tup{w}\right)
	\]
	
	\item For the quantifier rules we write $Q(\tup{x})$ to denote all occurrences of the free variables $\tup{x}$ in $Q$, and note that there may be no occurrences.  We assume the usual rules for substitution e.g. in $(\exists_R)$ that $\tup{t}$ is free for $\tup{x}$ in $Q$, and similarly for the other rules.
	
	\item Rule $(cons)$ is the Dialectica version of the usual consequence rules, where the pre- and postconditions can be weakened/strengthened with no bearing on the program. More precisely,
    we define the $\weha$-formula $P'\to_D P$ by
	\[
	\forall\tup{x},\tup{v}\left(\dt{P'}{\tup{x}}{\tup{v}}\to \dt{P}{\tup{x}}{\tup{v}}\right)
	\]
	i.e. the witnessing types of $P'$ and $P$ coincide under the Dialectica and implication is interpreted by the identity. 
	
	\item Finally, $\rec$ denotes the usual recursor whose defining axiom is given in Section \ref{sec:prelim:base}, while the backward recursor $\rec^\ast$ is defined (using the main recursor) as $\rec^\ast\tup{a}\tup{\alpha}\tup{y}x\tup{v}:=\tup{\beta}x$ for
	\[
	\tup{\beta}0:=\tup{v} \ \ \ \mbox{and} \ \ \ 
	\tup{\beta}(z+1):=\tup{\alpha}(x-z-1)(\rec\, \tup{a}\tup{y}(x-z-1))(\tup{\beta}z).
	\]

\end{itemize}


\begin{theorem}
\label{res:soundness}
Assuming that all elements of $\ax$ and $\rul$ are admissible in $\weha$, if $\hrd{P}{\tup{a}}{\tup{\alpha}}{Q}$ is derivable from the rules in Figure 1, then $\forall\tup{x},\tup{v}\left(\dt{P}{\tup{x}}{\tup{\alpha}\tup{x}\tup{v}}\to \dt{Q}{\tup{a}\tup{x}}{\tup{v}}\right)$ is provable in $\weha$.
\end{theorem}

The proof of Theorem \ref{res:soundness} is routine, and involves no ideas fundamentally different to those already present in the usual soundness proofs of Dialectica for $\weha$. Some of the more interesting rules are discussed in the appendix. We note that by restricting the rules of $\dhl$ so that the precondition is always purely existential, we obtain a simplified set of rules that act only on left-hand realizers, and very closely reflect traditional Hoare rules. These are also presented in the appendix, in Figure 9.

\paragraph*{Using $\dhl$}
\label{sec:dynamics:use}

\begin{figure*}[tp]
\centering
\begin{align*}
\boxed{
\begin{gathered}
\vlinf{}{p\vee_L'}
	{\hrd{Q\vee P}{\lambda c.\tilde{\tup{a}}\bar{c}}{\lambda c.\tilde{\tup{\beta}}\bar{c},\lambda c.\tilde{\tup{\alpha}}\bar{c}}{R}}
	{\hrd{P\vee Q}{\tup{a}}{\tup{\alpha},\tup{\beta}}{R}}
\qquad
\vlinf{}{\vee_R'}
	{\hrd{P}{\lambda\tup{x}.0,\tup{a},\tup{b}}{\tup{\alpha}_\pi}{Q\vee R}}
	{\hrd{P}{\tup{a}}{\tup{\alpha}}{Q}}
\\[2mm]
\vliinf{}{cond_L'}
	{\hrd{P\vee Q}{\lambda c,\tup{x},\tup{y}.\ite{c=0}{\tup{a}\tup{x}}{\tup{b}\tup{y}}}{\alpha_\pi,\beta_\pi}{R}}
	{\hrd{P}{\tup{a}}{\tup{\alpha}}{R}}
	{\hrd{Q}{\tup{b}}{\tup{\beta}}{R}}
\end{gathered}
}
\end{align*}
\label{fig:disjunction}
\caption{Rules for disjunction, derivable from those in Figure 1, which can be added to $\dhl$.}
\end{figure*}

Our system $\dhl$ is, above all, intended to be a useful system for constructing and describing realizing terms. It is certainly not a minimal system: Indeed, because $(cons)$ ranges over all implications of form $\to_D$ derivable in $\weha$, it can be used together with the quantifier introduction rules to derive all of the others. However, a human (or a machine) using the system would typically use $(cons)$ only in in cases where $P'\to_D P$ and $Q\to_D Q'$ are immediate (e.g. have no computational content), and would resort to the main rules for constructing complex terms.

We view the system as fundamentally extendible: The rules in Figure 1 provide the basic manipulations on terms required to show that the Dialectica is sound for $\weha$, but in practice, especially if we envisage $\dhl$ being incorporated into some proof assistant as a way of reasoning about Dialectica realizers, we would add a range of derivable rules for dealing with specific mathematical structures in proofs, along with new rules for characterising additional constructs added to System T. We give several examples of these in what follows. As a simple illustration, in Figure 2 we note several rules involving disjunction that are easily derivable from those for $\vee_b$ along with $(cons)$ and the fact that $P\vee Q\leftrightarrow_D \exists b(P\vee_b Q)$. 

It is now natural to ask: \emph{What purely logical system is mirrored by $\dhl$?} This can be quite interesting, especially so when we consider extensions to our rules later on. We first note that if we remove $(\epsilon_L)$, $(\epsilon_R)$, $(cons)$, $(ext)$ from those of Figure 1, and then replace the disjunction rules with those of Figure 2, these rules can be adapted to form a logical system in the language of $\weha$ by just suppressing the realizing term. For example, we can simplify the rule $(cond_L')$ as
\[
\vliinf{}{cond_L'}
	{P\vee Q\vdash R}
	{P\vdash R}
	{Q\vdash R}
\] 
and similarly for the others. Letting $\sys$ denote the resulting system, we have the following:
\begin{theorem}
\label{res:equivalence}
Whenever $P$ is provable in $\weha$, then $\top\vdash P$ is provable in $\sys$.
\end{theorem}
Again, the proof is straightforward, and a rough sketch of how we show this 
is provided in the appendix. Theorem \ref{res:equivalence} therefore represents an alternative soundness proof from the Dialectica interpretation of $\weha$, and another way of using $\dhl$: If $P$ is provable in $\weha$, we take an alternative derivation of $\top\vdash P$ in $\sys$ and adding back the hidden realizing terms we obtain a derivation of
\[
\hrd{\top}{\tup{a}}{-}{P}=\forall \tup{v}\dt{P}{\tup{a}}{\tup{v}}
\]
where $\tup{a}$ is constructed via the derivation (note that we can also force that the free variables of $\tup{a}$ are contained in those of $P$ by applying $(\exists_L)$ on all superfluous variables and then $(s_L)$ instantiated with zero terms). We also observe that if $P\to Q$ is any formula such that $P\leftrightarrow_D R\leftrightarrow_D Q$ for some suitable $R$, then $\hrd{P}{\lambda\tup{x}.\tup{x}}{\lambda\tup{x},\tup{v}.\tup{v}}{Q}$ is immediately derivable from the $(cons)$ rule and therefore $P\vdash Q$ can be added to $\sys$. In this way, we also regain that $\ac$, $\ip$ and $\markov$ are admissible by the Dialectica.

There are still three rules we have not discussed. While $(ext)$ is just an extensionality rule purely for reasoning about extracted terms, $(\epsilon_L)$ and $(\epsilon_R)$ are more interesting. These essentially allow us access to \emph{epsilon terms}: Informally speaking, we can add these rules to our purely logical system $\sys$ as
\[
\vlinf{}{\epsilon_R}
	{P_\forall\vdash Q(\tup{\epsilon})}
	{P_\forall\vdash \exists \tup{x}\, Q(\tup{x})}
\qquad
\vlinf{}{\epsilon_L}
	{P_\forall(\tup{\epsilon})\vdash Q_{qf}}
	{\forall \tup{x}\, P_\forall(\tup{x})\vdash Q_{qf}}
\]
where here we would add a countable set of epsilon terms $(\epsilon_i)$ to our language and use a fresh sequence of $\epsilon$-terms for each instance of $(\epsilon_L)$, $(\epsilon_R)$ in a derivation. The hidden realizing terms allow us to always assign concrete values to these $\epsilon$-terms, which can then be substituted in should any remain at the end of a derivation. However, some care is needed to set this up formally. Interestingly, the addition of $\epsilon$-terms to \emph{intuitionistic} logic is explored in \cite{viol:85:thesis}: There, doing this for arbitrary formulas does not give us a conservative extension of intuitionistic logic, as the full independence of premise axiom is then derivable\footnote{We are grateful to Cameron Allett for directing us to \cite{viol:85:thesis}.}. Our in-built restriction for $(\epsilon_R)$ that the left hand formula has purely universal Dialectica interpretation avoids this problem: We are then only able to derive independence of premise for \emph{universal} formulas, and this is admissible by Dialectica!

Finally, we note that by taking advantage of the infrastructure is already available via Oliva's unified approach to functional interpretations \cite{oliva:06:unifying}, it would not be hard to parametrise our own system $\dhl$ with abstract bounding relation $\sqsubset$ and thereby obtain both the traditional Dialectica, modified realizability and the Diller-Nahm interpretation \cite{diller:74:nahm} as instances. For the latter, the conditional rule $(cond_R)$ would then look like
\[
\vliinf{}{cond_R}
	{\hrd{P}{\tup{a},\tup{b}}{\lambda\tup{x},\tup{v},\tup{w}.\left(\tup{\alpha}\tup{x}\tup{v}\cup \tup{\beta}\tup{x}\tup{w}\right)}{Q\wedge R}}
	{\hrd{P}{\tup{a}}{\tup{\alpha}}{Q}}
	{\hrd{P}{\tup{b}}{\tup{\beta}}{R}}
\]
for $\cup$ a union operation on finite sets.

\subsection{A while rule for Dialectica, and its use in classical logic}
\label{sec:while}

We characterised $\dhl$ as being inspired by Hoare logic, but conspicuously absent so far are any kind of imperative control flow statements. To that end, we introduce a while construct and identify a corresponding rule that is sound in our base system. We argue that in the context of Dialectica, this while loop is well-suited to describing how programs extracted from classical logic implement ``learning algorithms''.

For the present paper, we restrict our attention to programs that are \emph{total}, so our while rule is specified relative to some wellfounded relation. We first introduce a decidable binary relation $\prec$ on objects of type $\tup{X}$, which can be modelled in $\weha$ as a function $\prec:\tup{X}\to\tup{X}\to \nat$. We then expand $\weha$ with a wellfounded induction rule $\ind{\prec}$ given by
\[
\vlinf{}{\ind{\prec}}
	{\forall \tup{x}\, A(\tup{x})}
	{\forall \tup{x}\left(\forall \tup{y}\prec \tup{x}\, A(\tup{y})\to A(\tup{x})\right)}
\]
where $A(\tup{x})$ ranges over arbitrary formulas. Now, for a quantifier-free formula $\phi(\tup{x}^{\tup{X}})$ and a term $\tup{a}:\tup{X}\to \tup{X}$, for any sequence of types $\tup{U}$ we add to the language of $\weha$ a while recursor $\whilerec{\prec}{\phi}{\tup{a}}:(\tup{X}\to \tup{U})\to(\tup{X}\to\tup{U}\to \tup{U})\to \tup{X}\to \tup{U}$ along with the schema
\[
(\phi(\tup{x})\to \tup{a}\tup{x}\prec\tup{x})\to\whilerec{\prec}{\phi}{\tup{a}}\tup{u}\tup{F}\tup{x}=_{\tup{U}}\ite{\phi(\tup{x})}{\tup{F}\tup{x}\left(\whilerec{\prec}{\phi}{\tup{a}}\tup{u}\tup{F}(\tup{a}\tup{x})\right)}{\tup{u}\tup{x}}
\]
where here the premise ensures that triggering the condition of the while loop causes a descent along $\prec$. We write $\weha+(\whilerecax{\prec})$ for $\weha$ extended with the while recursor and its axiom for all $\phi$ and $\tup{a}$. Now we define the forward while operator $\whiledo{\prec}{\phi}{\tup{a}}:\tup{X}\to\tup{X}$ as 
\[
\whiledo{\prec}{\phi}{\tup{a}}:=\whilerec{\prec}{\phi}{\tup{a}}(\lambda\tup{x}.\tup{x})(\lambda\tup{x},\tup{y}.\tup{y})
\]
and the backward operator $\whiledoback{\prec}{\phi}{\tup{a}}{\tup{\alpha}}:\tup{X}\to\tup{V}\to\tup{V}$ for $\tup{\alpha}:\tup{X}\to \tup{V}\to\tup{V}$ as
\[
\whiledoback{\prec}{\phi}{\tup{a}}{\tup{\alpha}}:=\whilerec{\prec}{\phi}{\tup{a}}(\lambda\tup{x},\tup{v}.\tup{v})(\lambda\tup{x},\tup{f},\tup{v}.\tup{\alpha}\tup{x}(\tup{f}\tup{v})).
\]
The following result, easily proven by expanding the definitions, guarantees that the forward while behaves indeed as a while loop. The backward direction is more subtle, and we discuss how it can be interpreted as part of a stateful program in Section \ref{sec:imperative}.
\begin{lemma}
\label{res:while:operator}
Whenever $\phi(\tup{x})\to \tup{a}\tup{x}\prec\tup{x}$ is provable in $\weha$, the following are provable in $\weha+(\whilerecax{\prec})$:
\begin{equation*}
\begin{aligned}
&(\whiledo{\prec}{\phi}{\tup{a}})\tup{x}=\ite{\phi(\tup{x})}{(\whiledo{\prec}{\phi}{\tup{a}})(\tup{a}\tup{x})}{\tup{x}},\\
&(\whiledoback{\prec}{\phi}{\tup{a}}{\tup{\alpha}})\tup{x}\tup{v}=\ite{\phi(\tup{x})}{\tup{\alpha}\tup{x}\left((\whiledoback{\prec}{\phi}{\tup{a}}{\tup{\alpha}})(\tup{a}\tup{x})\tup{v}\right)}{\tup{v}}.
\end{aligned}
\end{equation*}
\end{lemma}
We now define a corresponding rule for any decidable relation $\prec$, namely
\begin{align*}
\begin{gathered}
\vliinf{}{W_\prec}
	{\hrd{\exists\tup{x}\, P_\forall(\tup{x})}{\whiledo{\prec}{\phi}{\tup{a}}}{\whiledoback{\prec}{\phi}{\tup{a}}{\tup{\alpha}}}{\exists \tup{x}\left(P_\forall(\tup{x})\wedge \neg\phi(\tup{x})\right)}}
	{\hrd{\exists\tup{x}\left(P_\forall(\tup{x})\wedge\phi(\tup{x})\right)}{\tup{a}}{\tup{\alpha}}{\exists\tup{x}\, P_\forall(\tup{x})}}
	{\forall\tup{x}\left(\phi(\tup{x})\to \tup{a}\tup{x}\prec\tup{x}\right)}
\end{gathered}
\end{align*}
The rule $W_\prec$ is closely connected to the wellfounded while rule of Hoare logic, with the difference that the descent condition is now presented separately. The following result (proven in the appendix) connects the rule to the recursor, and bears some relationship to the interpretation of wellfounded induction in \cite{schwichtenberg:08:wellfounded}.
\begin{theorem}
\label{res:while:sound}
The while rule $(W_\prec)$ is admissible in $\weha+\ind{\prec}+(\whilerecax{\prec})$.
\end{theorem}

\paragraph*{Discussion: The interpretation of classical logic and its interaction with $\whilerecax{\prec}$}
\label{sec:classical:interpreting}

\begin{figure*}[tp]
\centering
\begin{align*}
\boxed{\small
\begin{gathered}
\vlinf{}{CP}
	{\hrd{\neg Q}{\tilde{\tup{\alpha}}}{\tilde{\tup{a}}}{\neg P}}
	{\hrd{P}{\tup{a}}{\tup{\alpha}}{\neg\neg Q}}
\qquad
\vlinf{}{N}
	{\hrd{P}{\tup{a}}{\tup{\alpha}}{\neg\neg Q}}
	{\hrd{P}{\tup{a}}{\tup{\alpha}}{\forall \tup{g}\exists\tup{u}\, \dt{Q}{\tup{u}}{\tup{g}\tup{u}}}}
\qquad
\vliinf{}{comp_\neg}
	{\hrd{P}{\lambda\tup{x}.\tup{b}(\tup{a}\tup{x}\tup{\beta})}{\lambda\tup{x}.\tup{\alpha}\tup{x}\tup{\beta}}{R_\exists}}
	{\hrd{P}{\tup{a}}{\tup{\alpha}}{\neg\neg Q}}
	{\hrd{Q}{\tup{b}}{\tup{\beta}}{R_\exists}}
\end{gathered}
}
\end{align*}
\caption{Example rules for manipulating classical proofs under the double negation translation.}
\label{fig:doubleneg}
\end{figure*}

It is well known that classical logic can be given a computational interpretation via the Dialectica by first carrying out a negative translation (see \cite[Chapter 10]{kohlenbach:08:book}), where the embedding of classical logic into intuitionistic makes use of a number of (semi-intuitionistic) laws governing negations. Expanding $\dhl$ with these laws as primitives would allow for a streamlined verification of programs extracted from classical proofs, and though detailed exploration of this is beyond the scope of the present paper, we give three examples of the kind of rules we have in mind, in Figure \ref{fig:doubleneg}. 

Here, $(CP)$ and $(N)$ are extremely useful (and reversible) symmetry rules that one often encounters when analysing classical proofs, while $(comp_\neg)$ is a version of composition which illustrates one way in which the negative translation composed with Dialectica eliminates nonconstructive lemmas when proving purely existential formulas, as we discuss below.

\begin{figure*}
\fbox{
\begin{minipage}{0.99\textwidth}
\scriptsize
$\begin{gathered}
\!\!\vlderivation{
\vlin{}{cons}
	{\hrd{\forall y(\forall z\prec y \neg\theta(z)\to\neg\theta(y))}{-}{\lambda x,g.a_gx}{\forall x\, \neg\theta(x)}}
	{\vlin{}{CP}
		{\hrd{\neg\exists y( \theta(y)\wedge \forall z\prec y \neg\theta(z))}{-}{\lambda x,g.a_gx}{\neg\exists x\, \theta(x)}}
		{
		\vlin{}{N}
			{\hrd{\exists x\, \theta(x)}{\lambda x,g.a_gx}{-}{\neg\neg\exists y( \theta(y)\wedge \forall z\prec y \neg\theta(z))}}
			{
			\vlin{}{\forall_R}
				{\hrd{\exists x\, \theta(x)}{\lambda x,g.a_gx}{-}{\forall g\exists y( \theta(y)\wedge \neg\phi_g(y))}}
				{
				\vlin{}{W_\prec}
			{\hrd{\exists x\, \theta(x)}{\whiledo{\prec}{\phi_g}{g}}{-}{\exists y( \theta(y)\wedge \neg\phi_g(y))}}
			{\vlhy{\vdots}}
				}
			}
		}
	}
}
\,\,
\vlderivation{
\vlin{}{\forall_R}
	{\hrd{\top}{\lambda v.b(a(fv)\beta)}{-}{\forall v\exists u\, \theta(v,u)}}
	{\vliin{}{\!\!comp_\neg}
		{\hrd{\top}{b(a(fv)\beta)}{-}{\exists u\, \theta(v,u)}}
		{
		\vlin{}{s_R}
			{\hrd{\top}{a(fv)}{-}{\neg\neg Q_v}}
			{
			\vlin{}{}
				{\hrd{\top}{a}{-}{\forall z\neg\neg \exists x\forall y \varphi(z,x,y)}}
				{\vlhy{\vdots}}
			}
		}
		{
		\vlin{}{}
			{\hrd{Q_v}{b}{\beta}{\exists u\, \theta(v,u)}}
			{\vlhy{\vdots}}
		}
	}
}
\end{gathered}$
\end{minipage}
}
\caption{Left: For $g:X\to X$ we use the abbreviation $\phi_g(x):= gx\prec x\wedge \theta(gx)$ and $a_g:=\whiledo{\prec}{\phi_g}{g}$. One completes the derivation by using $\hrd{\exists x(\theta(x)\wedge \phi_g(x))}{g}{-}{\exists y\, \theta(y)}$ and $\forall x(\phi_g(x)\to gx\prec x)$, which are easily provable. Right: we use the negative translation to obtain a constructive proof of the double negated statement that appears in the left branch.}
\label{fig:classic}
\end{figure*}

An example of how our while rule in particular interacts with these new principles for handling double negation is given on the left side of Figure \ref{fig:classic}, where we use it to define a realiser for double-negated version of the following \emph{minimum principle}:
\[
\exists x\, \theta(x)\vdash \exists x\left(\theta(x)\wedge\forall y\prec x\neg\theta(y)\right),
\]
where $\theta(x^X)$ is a quantifier-free formula and $\prec$ is a wellfounded relation on $X$ (we omit simple instances of $(ext)$ to improve readability). This represents a derivation, in a Hoare style system, of the kind of intuitive algorithm for the minimum principle already described in e.g. \cite[Section 7]{berardi-bezem-coquand:98:bbc}, and the imperative flavour of such constructions in relation to the Dialectica interpretation is also discussed in \cite{powell:16:learning}, where an explicit connection to learning algorithms is drawn. This general phenomenon of computing approximations to ideal objects via learning is present in almost all approaches to giving a computational meaning to classical reasoning: the epsilon calculus, modified realizability \cite{berger-schwichtenber:95:classical}, game semantics \cite{berardi-bezem-coquand:98:bbc}, approaches based on learning \cite{aschieri-berardi:10:interative}, the $\lambda\mu$-calculus \cite{parigot:92:lambdamu}, call-cc \cite{griffin:90:control}, classical realizability \cite{krivine:09:realizability, miquel:11:realizability} and others (see \cite{powell:19:epsilon} for a comparative study of some of these), and we propose that our approach via Hoare logic might help make explicit this intuition in the context of Dialectica.

Another example where we see our rules for classical logic in action is given by the frequently occurring proof pattern where an instance of a nonconstructive principle is eliminated by the Dialectica when used as a lemma in the proof of a $\forall\exists$-theorem. More precisely, suppose that we establish $\forall v\exists u\, \theta(v,u)$ by showing (constructively) that for any $v$ we have
$
\exists x\forall y\, \varphi(fv,x,y)\to \exists u\, \theta(v,u),
$
for some function $f$, but where the premise is an instance of a principle only provable nonconstructively; for instance, $x$ might be a minimal element in some set, a maximal ideal, or a point after which a sequence is $2^{-fv}$ stable. Abbreviating $\exists x\forall y\, \varphi(fv,x,y)$ with $Q_v$, the derivation on the right in Figure \ref{fig:classic} provides a basic template for eliminating this nonconstructive principle via $(comp_\neg)$, combining a witness $a$ for the \emph{classical} Dialectica interpretation of the latter with a realizing pair for the \emph{constructive} implication above. There are many cases where the witness $a$ for the double negated lemma could be built in a natural way via using our while loop, against some ``counterexample function'' provided by the backward witness $\beta$. It might then form a core component of a procedural-style program for computing a witness for $\forall v\exists u\, \theta(v,u)$.

\section{A procedural language for Dialectica}
\label{sec:imperative}

So far, we have put forward the argument that having a rich library of rules for Dialectica triples, including imperative constructs like the while loop, is beneficial for describing programs naturally arising from mathematics in a procedural way. 
In this section, we consider this aspect of our system in a more formal manner. For now we consider a restricted system that aligns closely with traditional Hoare logic, comprising a small set of rules for generating ``stateful'' programs. Because these rules arise from the Dialectica interpretation, the result is a novel procedural language with nonstandard interpretation. To be more precise, we define a language $\loopd$, representing a Dialectica-like extension of a toy procedural language $\loopi$, and provide both a Hoare logic and an operational semantics, where the latter captures the idea that
\[
\loopd = \loopi + \text{backpropagation},
\]
i.e.\ commands of the language consist of an ordinary forward part, together with a backward component that can be computed via a generalised form of \emph{backpropagation}. We anticipate that this operational semantics that could eventually be extended to our entire system $\dhl$.

In order to define $\loopd$ and its corresponding logic, we extend our ambient theory $\weha$ with two new \emph{abstract types} $S$ and $T$ representing two sorts of \emph{state}. The extension of $\weha$ with abstract types is standard in the proof mining literature (see \cite{kohlenbach:05:metatheorems}), where those types are typically used to represent abstract spaces from mathematics.
We then extend the terms $\weha$ by introducing: 

\begin{itemize}

\item two equality symbols $=_S$ and $=_T$, which we write as predicates but which will technically be terms $=_S:S\to S\to \nat$ and $=_T:T\to T\to \nat$ (indicating that equality between states is a decidable property). In addition, we extend all the usual constructs of $\weha$ e.g. quantifiers, recursors, lambda abstraction, so that they apply to our new states;

\item two sets of primitive command, the set $\comf$ of \emph{forward commands} of type $S\to S$, and the set $\comb$ of \emph{backward commands} of type $S\to T\to T$;

\item a set $\expr$ of boolean expressions, which are formally a set of decidable predicates on $S$ represented in $\weha$ as characteristic functions of type $S\to S\to \nat$;

\item a set $\rel$ of wellfounded relations on $S$.

\end{itemize}

We also assume that we have a set of purely universal axioms $\mathcal{U}$ that characterise the meaning of commands and expressions. We then define the set $\comm$ of \emph{commands} of $\loopd$ as follows, all of which represent new constants or term forming operations in $\weha$ for generating pairs of terms of types $S\to S,S\to T\to T$:
\[
C::= \loopskip \mid \dpair{c}{\gamma} \mid \seq{C}{C} \mid \loopite{\phi}{C}{C} \mid \loopwhiledo{\prec}{\phi}{C} 
\]
where $c\in\comf$, $\gamma\in\comb$, $\phi\in\expr$ and $\prec\in\rel$.
\begin{figure*}[t]
\small\centering
\[\begin{array}{|c|c|c|}
\hline
          \loopd     & (\_)^+       & (\_)^- \\
              \hline
    \loopskip & \lambda s.s & \lambda s,t.t \\[0.7mm]
    \dpair{c}{\gamma} & c & \gamma \\[0.7mm]
    \seq{C_1}{C_2} & C_2^+\circ C_1^+ & C_1^-\ast_{C_1^+} C_2^- \\[0.7mm]
    \loopite{\phi}{C_1}{C_2} & \lambda s.\ite{\phi(s)}{C^+_1s}{C^+_2s} & \lambda s,t.\ite{\phi(s)}{C^-_1st}{C^-_2st}\\[0.7mm]
    \loopwhiledo{\prec}{\phi}{C} & \whiledo{\prec}{\phi}{C^+} & \whiledoback{\prec}{\phi}{C^+}{C^-}\\
		\hline
\end{array}\]
\caption{Decomposition in $\weha$ of $\loopd$ terms in System T terms. Remember the Dialectica composition $\ast$ from Section~\ref{sec:dynamics} and the while constructors from Section \ref{sec:while}.}
\label{fig:LOOPD}
\end{figure*}
The new constructors come equipped with the defining axiom
\[
C=\dpair{C^+}{C^-},
\]
where $\dpair{}{}$ is a fixed representation of pairs in $\weha$ and $C^+:S\to S$ and $C^-:S\to T\to T$ are defined inductively over $\comm$ as in Figure~\ref{fig:LOOPD}. In particular, notice that by Lemma \ref{res:while:operator} we have that if $\neg \phi(s)$, then $(\loopwhiledo{\prec}{\phi}{C_1})s=(\loopskip) s$, and whenever $\phi(s)$ and $C^+s\prec s$, then
\[
(\loopwhiledo{\prec}{\phi}{C})s=(\seq{C}{\loopwhiledo{\prec}{\phi}{C}})s.
\]
We denote by $\weha+(\loopd)$ the extension of $\weha$ with all of the above (including the while recursor and well-founded induction axiom for all $\prec\in \rel$). 

\begin{figure*}
\fbox{
\begin{minipage}{0.99\textwidth}
\small\centering
$\begin{gathered}
\totalhoare{P}{\loopskip}{P}
\qquad
\vlinf{}{}
	{\totalhoare{P}{\dpair{c}{\gamma}}{Q}}
	{P(s,\gamma st)\to Q(cs,t)\in \ax}
\qquad
\vliiinf{}{}
	{\totalhoare{P'}{C}{Q'}}
	{P'\to P}
	{\totalhoare{P}{C}{Q}}
	{Q\to  Q'}
\\[2mm]
\,
\vliinf{}{}
	{\totalhoare{P}{\seq{C_1}{C_2}}{R}}
	{\totalhoare{P}{C_1}{Q}}
	{\totalhoare{Q}{C_2}{R}}
\qquad
\vliinf{}{}
	{\totalhoare{P\vee_\phi Q}{\loopite{\phi}{C_1}{C_2}}{R}}
	{\totalhoare{P\wedge\phi}{C_1}{R}}
	{\totalhoare{Q\wedge\neg\phi}{C_2}{R}}
\qquad
\vliinf{}{}
	{\totalhoare{P}{\loopwhiledo{\prec}{\phi}{C}}{P\wedge\neg\phi}}
	{\totalhoare{P\wedge\phi}{C}{P}}
	{\phi(s)\to C^+s\prec s}
\end{gathered}$
\end{minipage}
}
\caption{Hoare rules for Dialectica triples in $\loopd$. Notice that we use some abbreviations: $P\wedge \phi$ is sugar for the formula $P(s,t)\wedge \phi(s)$, and $P\vee_\phi Q$ for the formula $P(s,t)\vee_{\phi(s)} Q(s,t)$.}
\label{fig:HL_loopd}
\end{figure*}

Figure 5 characterises our commands as comprising a standard imperative forward part, with a ``Dialectica version'' of the same component as the backward part. A corresponding set of specification rules is now given in terms of our Dialectica triples in Figure 6: Here we use the more compact and suggestive abbreviation
\[
\totalhoare{P}{C}{Q}:=\hrd{\exists s\forall t\, P(s,t)}{C^+}{C^-}{\exists s\forall t\, Q(s,t)}=\forall s,t\left(P(s,C^-st)\to Q(C^+s,t)\right)
\]
and restrict our attention to Dialectica triples whose pre- and postconditions are formulas of the form $\exists s^S\forall t^T\, P(s,t)$ where $P(s,t)$ is quantifier-free, and $s,t$ are the only free variables of type $S,T$ that appear in them. In the second rule, $\ax$ represents some set of axioms that govern the primitive commands, which we assume are included in $\mathcal{U}$. All of these rules are either instances of the main rules from previous sections, or in the case of conditional, easily derivable, and so it is straightforward to prove the following:
\begin{theorem}
  All rules in Figure 6 are admissible in $\weha+(\loopd)$.  
\end{theorem}
We also note that Figure 6 only shows one set of rules corresponding to traditional Hoare logic: Plenty more are admissible (e.g. those for conjunction and disjunction).

\paragraph*{An operational semantics for $\loopd$ via backpropagation}
\label{sec:imperative:operational}

\begin{figure*}[tp]
\centering
\begin{align*}
\boxed{
\begin{gathered}
\textbf{Forward semantics}
\\[2mm]
\forwards{s}{\loopskip}{s}{\epsilon}{\epsilon}
\qquad
\forwards{s}{\dpair{c}{\gamma}}{cs}{[s]}{[\gamma]}
\qquad
\vliinf{}{}
	{\forwards{s}{\seq{C_1}{C_2}}{s''}{\sigma'::\sigma}{\Gamma'::\Gamma}}
	{\forwards{s}{C_1}{s'}{\sigma}{\Gamma}}
	{\forwards{s'}{C_2}{s''}{\sigma'}{\Gamma'}}
\\[2mm]
\vliinf{}{}
	{\forwards{s}{\loopite{\phi}{C_1}{C_2}}{s'}{\sigma}{\Gamma}}
	{\phi(s)}
	{\forwards{s}{C_1}{s'}{\sigma}{\Gamma}}
\qquad
\vliinf{}{}
	{\forwards{s}{\loopite{\phi}{C_1}{C_2}}{s'}{\sigma}{\Gamma}}
	{\neg\phi(s)}
	{\forwards{s}{C_2}{s'}{\sigma}{\Gamma}}
\\[2mm]
\vliiiinf{}{}
	{\forwards{s}{\loopwhiledo{\prec}{\phi}{C}}{s''}{\sigma'::\sigma}{\Gamma'::\Gamma}}
	{\phi(s)}
	{\forwards{s}{C}{s'}{\sigma}{\Gamma}}
	{s'\prec s}
	{\forwards{s'}{\loopwhiledo{\prec}{\phi}{C}}{s''}{\sigma'}{\Gamma'}}
\qquad
\vlinf{}{}
	{\forwards{s}{\loopwhiledo{\prec}{\phi}{C}}{s}{\epsilon}{\epsilon}}
	{\neg\phi(s)}
\\[2mm]
\textbf{Backward semantics}
\\[2mm]
\backward{\sigma}{\Gamma}{t}{\sigma}{\Gamma}{t}
\qquad
\backward{s::\sigma}{\gamma::\Gamma}{t}{\sigma}{\Gamma}{\gamma st}
\qquad
\vliinf{}{}
	{\backward{\sigma}{\Gamma}{t}{\sigma''}{\Gamma''}{t''}}
	{\backward{\sigma}{\Gamma}{t}{\sigma'}{\Gamma'}{t'}}
	{\backward{\sigma'}{\Gamma'}{t'}{\sigma''}{\Gamma''}{t''}}
\end{gathered}
}
\end{align*}
\caption{Operational semantics for $\loopd$.}
\end{figure*}

So far, $\loopd$ is just an extension of System T, with equational rules that describe the meaning of terms. We now endow terms of $\loopd$ with a big step operational semantics, to highlight how they can be interpreted as ``stateful'' programs. 
We first introduce some notation: For any type $X$ we let $X^\ast$ denote the type of finite sequences of elements of $X$ (while this is not formally part of $\weha$, it can easily be encoded in the system, though we omit details). For $\sigma,\tau:X^\ast$ we write $\tau::\sigma:X^\ast$ for the concatenation of two sequences, and for $x:X$ similarly write $x::\sigma$ for concatenation with a single element. We write $\epsilon$ for the empty list (of any type). Our semantics comprises two components: A forward relation and a backward relation, which we write 
\[
\forwards{s}{C}{s'}{\sigma}{\Gamma} \ \ \ \mbox{and} \ \ \ \backward{\sigma}{\Gamma}{t}{\sigma'}{\Gamma'}{t'}
\]
for $s,s':S$, $t,t':T$, $C\in\comm$, $\sigma,\sigma':S^\ast$ and $\Gamma,\Gamma':(S\to T\to T)^\ast$. Rules for the semantics are given in Figure 7. If we ignore the stacks, our forward semantics is isomorphic to a standard operational semantics for an imperative language. It is not difficult to show that semantics is sound with respect to the denotational semantics of $\weha+(\loopd)$:

\begin{theorem}
\label{res:operational}
Let $C$ be an arbitrary command in $\loopd$, where we restrict the formation of $\loopwhiledo{\prec}{\phi}{C_1}$ to commands such that $\forall s(\phi(s)\to C_1^+s\prec s)$ is provable in $\weha$. Then for any $s$ there exist $s'$, $\sigma$ and $\Gamma$ such that
\[
\forwards{s}{C}{s'}{\sigma}{\Gamma}
\]
where $s'=C^+s$ in $\weha$, and such that for any $t:T$, $\sigma_0$ and $\Gamma_0$ there exists $t'$ such that
\[
\backward{\sigma::\sigma_0}{\Gamma::\Gamma_0}{t}{\sigma_0}{\Gamma_0}{t'}
\]
where $t'=C^-st$ in $\weha$.  
\end{theorem} 
A proof of Theorem \ref{res:operational} can be found in the appendix. Our operational semantics captures the idea that to evaluate a command $C$ in state $s$, we perform a forward run
\[
\forwards{s}{C}{s_k}{[s_{k-1},\ldots,s_1,s]}{[\gamma_k,\ldots,\gamma_1]}
\]
where for each $j=1,\ldots,k$ we have $s_j=c_js_{j-1}$ for $\dpair{c_j}{\gamma_j}$ a primitive command, and $s_k=(c_k\circ\ldots\circ c_1)s=C^+s$. In other words, each time we are confronted with a primitive command $\dpair{c_j}{\gamma_j}$ in some intermediate state $s_{j-1}$, we perform the forward component but push both the intermediate state $s_{j-1}$ and the backward command $\gamma_j$ onto stacks. The backward run then just pops the intermediate states and commands, with
\[
\backward{[s_{k-1},\ldots,s_1,s]}{[\gamma_k,\ldots,\gamma_1]}{t}{\epsilon}{\epsilon}{(\gamma_1s\circ\ldots\circ\gamma_ks_{k-1})t}
\]
where $(\gamma_1s\circ\ldots\circ\gamma_ks_{k-1})t=C^-st$ 
In this way, programs in $\loopd$ combine a generalised backpropagation algorithm, where in addition to composing functions we can also perform conditionals and loops. The fact that Dialectica is tightly related with manipulations of stacks in a ``backward way'' has been discovered in \cite{pedrot:14:functional}, therefore it is not a surprise that stacks also play an essential role in our framework.

\paragraph*{Discussion: On automatic differentiation}

The basic connection between Dialectica and backpropagation is not new: In particular, it was explored in the specific context of differentiation in \cite{kerjean:pedrot:24:delta}. We can phrase it within our framework as well:
Suppose that our state $S$ represent some space on which we can define functions $c:S \to S$ and we have the notion of a differentials $\diff{c}{s}:S\to S$ at points $s:S$. Suppose that we can also express the reverse differentials $\rdiff{c}{s}:T\to T$ on the dual space $T$, defined by $\rdiff{c}{s} t:=t\circ \diff{c}{s}$. Then for a pair $c_1,c_2:S\to S$ of differentiable functions, it is well-known that we have the transpose version of the chain rule:
\[
\rdiff{c_1}{s}\circ\left(\rdiff{c_2}{c_1(s)}\right)
\]
or, to put it another way,
\[
\dpair{c_2\circ c_1}{\rdiff{c_2\circ c_1}{(\cdot)}}=\seq{\dpair{c_1}{\rdiff{c_1}{(\cdot)}}}{\dpair{c_2}{\rdiff{c_2}{(\cdot)}}}.
\]
Therefore, if our primitive commands consist of pairs $\dpair{c}{\gamma}$ with the property that $\gamma s=\rdiff{c}{s}$, this property is preserved under composition of commands, and the resulting semantics (restricted to composition) gives a version of the traditional backpropagation algorithm.

It should be stressed that our approach is not intended to equate the investigations around Dialectica and differentiation as done in \cite{kerjean:pedrot:24:delta}, where among other things a precise translation into the differential $\lambda$-calculus 
is provided. In a certain sense, our presentation is orthogonal, presenting a general backpropagation procedure connected to imperative commands. 
Actually, for commands of the form $\dpair{c}{\rdiff{c}{(\cdot)}}$ to have any real place within our framework, they have to connect to the overarching logic: For what kind of predicates $P,Q$ on spaces $S$ and their duals $T$ is it naturally the case that
$
\forall s^S,t^{S\to R}\left(P(s,t\circ \diff{c}{s})\to Q(cs,t)\right)
$
holds? This is currently unclear to the authors.

\section{Conclusion and Future Work}
\label{sec:conc}

\noindent\textbf{Towards a genuine imperative language, concurrency.} 
So far, our presentation of $\loopd$ gives no concrete \emph{data structure} on the states $S,T$ above, and our primitive commands $c,\gamma$ are completely abstract. 
This is the only missing step in order transform $\loopd$ into a full imperative language with backpropagation. This can be achieved, for example, by introducing, as constants in our formal system, a countable collection of state variables $\x_1,\x_2,\ldots$ modeling a memory heap for the state $S$, which would be now seen as a partial function mapping state variables to $\nat$, and similarly we could ask for a further memory heap for the state $T$. In a similar way we could introduce arithmetic expressions, and include an assignment operation in our primitive commands. However, we leave the details to future work.

\begin{figure*}
\fbox{
\begin{minipage}{0.99\textwidth}
\small\centering
$\begin{gathered}
     \vliinf{}{}
	{P_1\ast P_2\, \dpair{\tup{a}}{\tup{\alpha}}\, || \, \dpair{\tup{b}}{\tup{\beta}}\, Q_1\ast Q_2}
	{\hrd{P_1}{\tup{a}}{\tup{\alpha}}{Q_1}}
	{\hrd{P_2}{\tup{b}}{\tup{\beta}}{Q_2}}
    \qquad\qquad
\vlinf{}{}
	{P\ast R\, \dpair{\tup{a}}{\tup{\alpha}}\, || \, \loopskip\, \, Q\ast R}
	{\hrd{P}{\tup{a}}{\tup{\alpha}}{Q}}
\end{gathered}$
\end{minipage}
}
\caption{Frame-like rules admissible in $\dhl$, $P_1\ast P_2\, \dpair{\tup{a}}{\tup{\alpha}}\, || \, \dpair{\tup{b}}{\tup{\beta}}\, Q_1\ast Q_2$ is an abbreviation for $P_1\wedge P_2\, \dpair{\tup{a},\tup{b}}{\tup{\alpha},\tup{\beta}}\, Q_1\wedge Q_2$ \emph{under the condition that $\tup{a},\tup{b}$ and $\tup{\alpha},\tup{\beta}$ operate on different variables}.}
\label{fig:concurrency}
\end{figure*}

Another feature that seems to naturally appear in our framework is \emph{concurrency}, where in particular the rules we give in Figure \ref{fig:concurrency} are admissible in $\dhl$. The interpretation of the conclusion of the left hand rule is
\[
\forall\tup{x},\tup{y},\tup{u},\tup{v}\left(\dt{P_1}{\tup{x}}{\tup{\alpha}\tup{x}\tup{u}}\wedge\dt{P_2}{\tup{y}}{\tup{\beta}\tup{y}\tup{v}}\to \dt{Q_1}{\tup{a}\tup{x}}{\tup{u}}\wedge \dt{Q_2}{\tup{b}\tup{y}}{\tup{v}}\right)
\]
i.e.\ Dialectica realisers act on \emph{independent} variables. Accordingly, we also have the basic variant of the frame rule in the right of Figure \ref{fig:concurrency}. A proper incorporation of concurrency into Dialectica would be extremely interesting, making the kind local reasoning already implicit in Dialectica more formal.
Strictly related with concurrency would be to extend Dialectica to bunched logic \cite{ohearn-pym:99:bunched}. In fact, success here through our approach as Hoare Logic, would lead to further developments in the direction of its extension to separation logic \cite{reynolds:02:separation} (or intermediate logics \cite{ciabattoni:etal:20:concurrent}). The existence of Dialectica interpretations of linear logic \cite{oliva:10:linear:intuitionistic} demonstrate the possibility of handling the computational content of substructural logics with Dialectica.
\medskip


\noindent\textbf{Case studies in probability.} Over the last few years, proof mining has been rapidly expanding to probability theory. An entirely different route for expansion is then to investigate Dialectica for probabilistic programs, in order to make formal the extraction processes behind those new works. 
Here we could potentially combine our system for Dialectica with variants of Hoare logic for probabilistic programs,~e.g. \cite{hartog:etal:probabilistic:hoare}.
We anticipate that this is a fertile new territory for interesting algorithms: Here, iterative trial-and-error algorithms seem fundamental, with several different forms of probabilistic convergence represented computationally in terms of learning procedures, which, informally speaking, test elements of a stochastic processes until a region is found which is locally stable with some sufficiently high probability (see \cite{neri-pischke-powell:25:learnability} for a more detailed account of this phenomenon in the specific context of stochastic convergence, and \cite{neri-powell:25:martingale,neri-powell:pp:rs} for examples of recent, relevant case studies, which in turn build on earlier work in ergodic theory \cite{avigad-gerhardy-towsner:10:local}). We note that even elementary facts from probability have resulted in algorithms of extreme complexity, such as the analysis of Egorov's theorem in \cite{avigad-dean-rute:12:dominated}, and propose that the procedural paradigm explored in this paper might be well suited to describing and simplifying such algorithms.\medskip

\noindent\textbf{Operational semantics of Dialectica.} Finally, 
it would be interesting to try to extend the operational semantics of $\loopd$ terms to the more general class of realisers handled by the full system $\dhl$. Here, of course, we have to contend with full functional programming, but we conjecture that e.g.\ through a sensible use of monads one could potentially characterise general Dialectica realisers more dynamically, by describing in more detail how the forward and backward directions interact.
This perspective brings us closer to machine-based interpretations of proofs such as classical realizability \cite{krivine:09:realizability}, whose applicability to witness extraction from classical proofs as discussed in \cite{miquel:11:realizability} is undoubtedly relevant, given the association with negative translations. 
The starting point would surely be the 
already mentioned 
\cite{pedrot:14:functional}, where a fundamental connection between Dialectica and the KAM is given.

\bibliography{biblio,tpbiblio}

\section{Appendix}

All proofs contained in the appendix are routine, involving a standard structural induction with multiple cases.

\begin{proof}[Proof of Theorem~\ref{res:soundness}] The proof follows the standard soundness theorem for the Dialectica, so we only give details of representative and interesting cases.
\begin{itemize}

\item The axioms are admissible by definition. 

\item For $(p\wedge_R)$, the premise is
\[
\forall \tup{x},\tup{v},\tup{w}\left(\dt{P}{\tup{x}}{\tup{\alpha}\tup{x}\tup{v}\tup{w}}\to \dt{Q}{\tup{a}\tup{x}}{\tup{v}}\wedge \dt{R}{\tup{b}\tup{y}}{\tup{w}}\right)
\]
and therefore we have
\[
\forall \tup{x},\tup{w},\tup{w}\left(\dt{P}{\tup{x}}{\tilde{\tup{\alpha}}\tup{x}\tup{w}\tup{v}}\to \dt{R}{\tup{b}\tup{y}}{\tup{w}}\wedge \dt{Q}{\tup{a}\tup{x}}{\tup{v}}\right)
\]
for $\tilde{\tup{\alpha}}\tup{x}\tup{w}\tup{v}:=\tup{\alpha}\tup{x}\tup{v}\tup{w}$. All other basic actions are proved similarly.

\item $(cond_R)$ is the more interesting of the conditionals. Here, we have
\[
\forall\tup{x},\tup{v}\left(\dt{P}{\tup{x}}{\tup{\alpha}\tup{x}\tup{v}}\to \dt{Q}{\tup{a}\tup{x}}{\tup{v}}\right)
\]
and 
\[
\forall\tup{x},\tup{w}\left(\dt{P}{\tup{x}}{\tup{\beta}\tup{x}\tup{w}}\to \dt{R}{\tup{b}\tup{x}}{\tup{w}}\right).
\]
Define $\tup{\gamma}:=\lambda\tup{x},\tup{v},\tup{w}.\ite{\dt{P}{\tup{x}}{\tup{\alpha}\tup{x}\tup{v}}}{\tup{\beta}\tup{x}\tup{w}}{\tup{\alpha}\tup{x}\tup{v}}$. We need to prove that
\[
\forall\tup{x},\tup{v},\tup{w}\left(\dt{P}{\tup{x}}{\tup{\gamma}\tup{x}\tup{v}\tup{w}}\to \dt{Q}{\tup{a}\tup{x}}{\tup{v}}\wedge \dt{R}{\tup{b}\tup{x}}{\tup{w}}\right)
\]
There are two possibilities: Fixing $\tup{x},\tup{v},\tup{w}$, either $\dt{P}{\tup{x}}{\tup{\alpha}\tup{x}\tup{v}}$ holds, in which case 
\[
\dt{P}{\tup{x}}{\tup{\gamma}\tup{x}\tup{v}\tup{w}}\to \dt{P}{\tup{x}}{\tup{\beta}\tup{x}\tup{w}}\to \dt{P}{\tup{x}}{\tup{\alpha}\tup{x}\tup{v}}\wedge \dt{P}{\tup{x}}{\tup{\beta}\tup{x}\tup{w}}
\]
and so the conclusion is true, or $\neg \dt{P}{\tup{x}}{\tup{\alpha}\tup{x}\tup{v}}$, and then $\dt{P}{\tup{x}}{\tup{\gamma}\tup{x}\tup{v}\tup{w}}\to \dt{P}{\tup{x}}{\tup{\alpha}\tup{x}\tup{v}}\to \bot$, and so the result automatically follows by ex-falso-quodlibet.

\item $(imp)$ and $(exp)$ are immediate.

\item $(comp)$ is standard and fundamental to Dialectica: If
\[
\forall\tup{x},\tup{v}\left(\dt{P}{\tup{x}}{\tup{\alpha}\tup{x}\tup{v}}\to \dt{Q}{\tup{a}\tup{x}}{\tup{v}}\right)
\]
and 
\[
\forall\tup{u},\tup{w}\left(\dt{Q}{\tup{u}}{\tup{\beta}\tup{u}\tup{w}}\to \dt{R}{\tup{b}\tup{u}}{\tup{w}}\right)
\]
then fixing $\tup{x},\tup{w}$, setting $\tup{u}:=\tup{a}\tup{x}$ and $\tup{v}:=\tup{\beta}\tup(\tup{a}\tup{x})\tup{w}$ gives the desired result.

\item The first four quantifier rules are also standard, though a little care is needed if we allow them to apply to tuples. The most involved in $(\forall_R)$. If
\[
\forall\tup{y},\tup{v}\left(\dt{P}{\tup{y}}{\tup{\alpha}\tup{y}\tup{v}}\to \dt{Q(\tup{x})}{\tup{a}\tup{y}}{\tup{v}}\right)
\]
then
\[
\forall\tup{y},\tup{v}\left(\dt{P}{\tup{y}}{\tup{\alpha}\tup{y}\tup{v}}\to \dt{Q(x_1,\ldots,x_n)}{(\lambda x_n.\tup{a}\tup{y})x_n}{\tup{v}}\right)
\]
which is just
\[
\forall\tup{y},\tup{v}\left(\dt{P}{\tup{y}}{\tup{\alpha}\tup{y}\tup{v}}\to \dt{\forall x_n\, Q(x_1,\ldots,x_n)}{\lambda x_n.\tup{a}\tup{y}}{x_n,\tup{v}}\right)
\]
and continuing for the rest of the tuple we obtain
\[
\forall\tup{y},\tup{v}\left(\dt{P}{\tup{y}}{\tup{\alpha}\tup{y}\tup{v}}\to \dt{\forall \tup{x}\, Q(\tup{x})}{\lambda \tup{x}.\tup{a}\tup{y}}{\tup{x},\tup{v}}\right)
\]
which is just
\[
\forall\tup{y},\tup{v}\left(\dt{P}{\tup{y}}{(\lambda\tup{y},\tup{x}.\tup{\alpha}\tup{y})\tup{y}\tup{x}\tup{v}}\to \dt{\forall \tup{x}\, Q(\tup{x})}{(\lambda \tup{y},\tup{x}.\tup{a}\tup{y})\tup{y}}{\tup{x},\tup{v}}\right)
\]
where here we note that the condition $\tup{x}$ not free in $P$ ensures that the free variables of $\dt{P\to \forall\tup{x}\, Q(\tup{x})}{\tup{f},\tup{F}}{\tup{y},\tup{x},\tup{v}}$ are the same as those of $P\to \forall\tup{x}\, Q(\tup{x})$.

\item $(s_L)$ and $(s_R)$ follow in a straightforward way from the usual quantifier rules, and the conclusion and premise is identical for $(\epsilon_R)$ and $(\epsilon_L)$.

\item $(cons)$ is immediate, and $(exp)$ follows from the rule of extensionality in $\weha$. 

\item As usual, $(ind)$ is proven by induction. We have
\[
\forall\tup{y},\tup{v}\left(\dt{P(x)}{\tup{y}}{\tup{\alpha}(x)\tup{y}\tup{v}}\to \dt{P(x+1)}{\tup{a}(x)\tup{y}}{\tup{v}}\right)
\]
for all $x:\nat$, so fixing $\tup{y}$, $\tup{v}$ and defining $\tup{b}:=\rec\,\tup{a}\tup{y}$ and $\tup{\beta}$ as in Section \ref{sec:dynamics}, we prove by induction that
\[
\dt{P(x-z)}{\tup{b}(x-z)}{\tup{\beta} z}\to \dt{P(x)}{\tup{b}x}{\tup{v}}
\]
for $z=0,\ldots,x$. The base case is immediate, and for the induction step we use that $\dt{P(x-z-1)}{\tup{b}(x-z-1)}{\tup{\beta}(z+1)}$ is equivalent to
\[
\dt{P(x-z-1)}{\tup{b}(x-z-1)}{\tup{\alpha}(x-z-1)(\tup{b}(x-z-1))(\tup{\beta}z)}
\]
which by the premise of the rule allows us to obtain
\[
\dt{P(x-z)}{\tup{a}(x-z-1)(\tup{b}(x-z-1))}{\tup{\beta}z}
\]
which is just $\dt{P(x-z)}{\tup{b}(x-z)}{\tup{\beta}z}$. Thus we can apply the induction hypothesis. For $x:=z$ we then have
\[
\dt{P(0)}{\tup{y}}{\tup{\beta} x}\to \dt{P(x)}{\tup{b}x}{\tup{v}}
\]
and the result follows by definition.
\end{itemize}
The remaining cases are straightforward.
\end{proof}

\begin{figure*}
\fbox{
\begin{minipage}{0.94\textwidth}
\scriptsize
$\begin{gathered}
\textbf{Propositional rules: Axioms, basic actions, conditionals, switching, composition.}
\\[2mm]
\hrd{\bot}{\tup{a}}{-}{P}
\quad
\hrd{P_\exists}{-}{\tup{\alpha}}{\top}
\quad
\hrd{P_\exists}{\lambda\tup{x}.\tup{x}}{-}{P_\exists}
\quad
\vlinf{}{}
	{\hrd{P_\exists}{-}{-}{Q_\forall}}
	{P_\exists\to Q_\forall\in\ax}
\qquad
\vlinf{}{}
	{\hrd{P'_\exists}{-}{-}{Q'_\forall}}
	{\hrd{P_\exists}{-}{-}{Q_\forall}}
\text{\,  for $\frac{P_\exists\to Q_\forall}{P'_\exists\to Q'_\forall}\in\rul$}
\\[2mm]
\vlinf{}{p\wedge_R}
	{\hrd{P_\exists}{\tup{b},\tup{a}}{-}{R\wedge Q}}
	{\hrd{P_\exists}{\tup{a},\tup{b}}{-}{Q\wedge R}}
\qquad
\vlinf{}{p\wedge_L}
	{\hrd{Q_\exists\wedge P_\exists}{\tilde{\tup{a}}}{-}{R}}
	{\hrd{P_\exists\wedge Q_\exists}{\tup{a}}{-}{R}}
\qquad
\vlinf{}{p\vee_R}
	{\hrd{P_\exists}{\tup{b},\tup{a}}{-}{R\vee_{\bar{c}} Q}}
	{\hrd{P_\exists}{\tup{a},\tup{b}}{-}{Q\vee_c R}}
\qquad
\vlinf{}{p\vee_L}
	{\hrd{Q_\exists\vee_{\bar{c}} P_\exists}{\tilde{\tup{a}}}{-}{R}}
	{\hrd{P_\exists\vee_c Q_\exists}{\tup{a}}{-}{R}}
\\[2mm]
\vlinf{}{\vee_R}
	{\hrd{P_\exists}{\tup{a},\tup{b}}{-}{Q\vee_0 R}}
	{\hrd{P_\exists}{\tup{a}}{-}{Q}}
\qquad
\vlinf{}{\wedge_L}
	{\hrd{P_\exists\wedge R_\exists}{\tup{a}_\pi}{-}{Q}}
	{\hrd{P_\exists}{\tup{a}}{-}{Q}}
\qquad
\vlinf{}{\wedge_R}
	{\hrd{P_\exists}{\tup{a}}{-}{Q}}
	{\hrd{P_\exists}{\tup{a},\tup{b}}{-}{Q\wedge R}}
\qquad
\vlinf{}{\vee_L}
	{\hrd{P_\exists}{\tup{a}_p}{-}{Q}}
	{\hrd{P_\exists\vee_0 R_\exists}{\tup{a}}{-}{Q}}
\\[2mm]
\vliinf{}{cond_L}
	{\hrd{P_\exists\vee_\phi Q_\exists}{\lambda\tup{x},\tup{y}.\ite{\phi}{\tup{a}\tup{x}}{\tup{b}\tup{y}}}{-}{R}}
	{\hrd{P_\exists\wedge\phi}{\tup{a}}{-}{R}}
	{\hrd{Q_\exists\wedge\neg\phi}{\tup{b}}{-}{R}}
\qquad
\vliinf{}{cond_R}
	{\hrd{P_\exists}{\tup{a},\tup{b}}{-}{Q\wedge R}}
	{\hrd{P_\exists}{\tup{a}}{-}{Q}}
	{\hrd{P_\exists}{\tup{b}}{-}{R}}
\\[2mm]
\vlinf{}{imp}
	{\hrd{P_\exists\wedge Q_\exists}{\tup{a}}{-}{R}}
	{\hrd{P_\exists}{\tup{a}}{-}{Q_\exists\to R}}
\qquad
\vlinf{}{exp}
	{\hrd{P_\exists}{\tup{a}}{-}{Q_\exists\to R}}
	{\hrd{P_\exists\wedge Q_\exists}{\tup{a}}{-}{R}}
\qquad
\vliinf{}{comp}
	{\hrd{P_\exists}{\tup{b}\circ\tup{a}}{-}{R}}
	{\hrd{P_\exists}{\tup{a}}{-}{Q_\exists}}
	{\hrd{Q_\exists}{\tup{b}}{-}{R}}
\qquad
\\[2mm]
\textbf{Quantifier rules: Term introduction, $\lambda$-abstraction and application, epsilon terms.}
\\[2mm]
\vlinf{}{\exists_R}
	{\hrd{P_\exists}{\lambda\_.t,\tup{a}}{-}{\exists \tup{x}\, Q(\tup{x})}}
	{\hrd{P_\exists}{\tup{a}}{-}{Q(\tup{t})}}
\\[2mm]
\vlinf{}{\exists_L}
	{\hrd{\exists \tup{x}\, P_\exists(\tup{x})}{\lambda \tup{x}.\tup{a}}{-}{Q}}
	{\hrd{P_\exists(\tup{x})}{\tup{a}}{-}{Q}}
\qquad
\vlinf{}{\forall_R}
	{\hrd{P_\exists}{\lambda\tup{y},\tup{x}.\tup{a}\tup{y}}{-}{\forall\tup{x}\, Q(\tup{x})}}
	{\hrd{P_\exists}{\tup{a}}{-}{Q(\tup{x})}}
\qquad
\text{$\tup{x}$ not free in $Q$ resp. $P$ for $\exists_L$ resp. $\forall_R$}
\\[2mm]
\vlinf{}{s_L}
	{\hrd{P_\exists(\tup{t})}{\tup{a}\tup{t}}{-}{Q}}
	{\exists \tup{x}\, \hrd{P_\exists(\tup{x})}{\tup{a}}{-}{Q}}
\qquad
\vlinf{}{s_R}
	{\hrd{P_\exists}{\lambda\tup{y}.\tup{a}\tup{y}\tup{t}}{-}{Q(\tup{t})}}
	{\hrd{P_\exists}{\tup{a}}{-}{\forall \tup{x}\, Q(\tup{x})}}
\qquad
\vlinf{}{\epsilon_R}
	{\hrd{P_{qf}}{\tup{b}}{-}{Q(\tup{a})}}
	{\hrd{P_{qf}}{\tup{a},\tup{b}}{-}{\exists \tup{x}\, Q(\tup{x})}}
\\[2mm]
\textbf{Consequence, extensionality, and induction/recursion}
\\[2mm]
\vliiinf{}{cons}
	{\hrd{P'_\exists}{\tup{a}}{-}{Q'}}
	{P'_\exists\to_D P_\exists}
	{\hrd{P_\exists}{\tup{a}}{-}{Q}}
	{Q\to_D Q'}
\qquad
\vliinf{}{ext}
	{\hrd{P_\exists}{\tup{b}}{\tup{\beta}}{Q}}
	{\hrd{P_\exists}{\tup{a}}{\tup{\alpha}}{Q}}
	{\tup{a},\tup{\alpha}=\tup{b},\tup{\beta}}
\qquad
\vlinf{}{\!ind}
	{\hrd{P_\exists(0)}{\rec\, {\tup{a}}}{-}{\forall x\, P_\exists(x)}}
	{\hrd{P_\exists(x)}{\tup{a}x}{-}{P_\exists(x+1)}}
\end{gathered}$
\end{minipage}
}
\caption{Simplified rules for Dialectica triples with empty backward realiser.}
\label{fig:simp}
\end{figure*}

\begin{proof}[Proof of Theorem \ref{res:equivalence}]
We refer to the axiomatisation given in \cite[Section 3]{kohlenbach:08:book}. We first note that, in our system, provability of $\top\vdash P\to Q$ is equivalent to provability of $P\vdash Q$. With that in mind, for the axioms of intuitionistic logic, both contraction axioms follow from the conditional rules, while weakening, permutation, and ex falso quodlibet are clearly derivable. The quantifier axioms follow from $(s_R)$ and $(s_L)$. Both modus ponens and syllogism are instances of $(comp)$, where for the former we note that if $\top\vdash P$ and $\top\vdash P\to Q$, then also $P\vdash Q$, and thus $\top\vdash Q$. Exportation and importation are identical in both systems, while expansion is provable using a combination of the rules for $\vee$. The quantifier rules are just $(\exists_L)$ and $(\forall_R)$. For arithmetic: We assume that the axioms and rules for equality and System T are included in $\ax$, and the quantifier-free rules of extensionality included in $\rul$, so all of these are then provable in our system. Replacing the induction axioms with the equivalent rule, it is not hard to show that the latter is derivable from $(ind)$.
\end{proof}

\begin{proof}[Proof of Theorem \ref{res:while:sound}]
Suppose that the premises of the rule hold, and so in particular by the left hand premise
\[
\forall\tup{x},\tup{v}\left(\dt{P_\forall(\tup{x})}{}{\tup{\alpha}\tup{x}\tup{v}}\wedge\phi(\tup{x})\to \dt{P_\forall(\tup{a}\tup{x})}{}{\tup{v}}\right)
\]
Writing $\tup{b}:=\whiledo{\prec}{\phi}{\tup{a}}$ and $\tup{\beta}:=\whiledoback{\prec}{\phi}{\tup{a}}{\tup{\alpha}}$, we are done if we can prove $\forall\tup{x}\, A(\tup{x})$ for
\[
A(\tup{x}):=\forall\tup{v}\left(\dt{P_\forall(\tup{x})}{}{\tup{\beta}\tup{x}\tup{v}}\to \dt{P_\forall(\tup{b}\tup{x})}{}{\tup{v}}\wedge\neg\phi(\tup{b}\tup{x})\right)
\]
We do this using the wellfounded induction rule $\ind{\prec}$. So fixing $\tup{x}$, we assume that $A(\tup{y})$ holds for all $\tup{y}\prec\tup{x}$. To prove $A(\tup{x})$, we use Lemma \ref{res:while:operator}, which is applicable for any $\tup{x}$ thanks to the right hand premise of $(W_\prec)$. There are two cases to consider. If $\neg\phi(\tup{x})$ then $\tup{b}\tup{x}=\tup{x}$ and $\tup{\beta}\tup{x}=\lambda\tup{v}.\tup{v}$, and then $A(\tup{x})$ becomes
\[
\forall\tup{v}\left(\dt{P_\forall(\tup{x})}{}{\tup{v}}\to \dt{P_\forall(\tup{x})}{}{\tup{v}}\wedge\neg\phi(\tup{x})\right)
\]
which is provable in this case. On the other hand, if $\phi(\tup{x})$ then we have $\tup{b}\tup{x}=\tup{b}(\tup{a}\tup{x})$, $\tup{\beta}\tup{x}=\lambda\tup{v}.\tup{\alpha}\tup{x}(\tup{\beta}(\tup{a}\tup{x})\tup{v})$ and so $A(\tup{x})$ becomes
\[
\forall\tup{v}\left(\dt{P_\forall(\tup{x})}{}{\tup{\alpha}\tup{x}(\tup{\beta}(\tup{a}\tup{x})\tup{v})}\to \dt{P_\forall(\tup{b}(\tup{a}\tup{x}))}{}{\tup{v}}\wedge\neg\phi(\tup{b}(\tup{a}\tup{x}))\right)
\]
But from $\phi(\tup{x})$ we also obtain $\tup{a}\tup{x}\prec\tup{x}$, and so by the induction hypothesis we can assume
\[
\forall\tup{v}\left(\dt{P_\forall(\tup{a}\tup{x})}{}{\tup{\beta}(\tup{a}\tup{x})\tup{v}}\to \dt{P_\forall(\tup{b}(\tup{a}\tup{x}))}{}{\tup{v}}\wedge\neg\phi(\tup{b}(\tup{a}\tup{x}))\right)
\]
It suffices therefore to show that
\[
\forall\tup{v}\left(\dt{P_\forall(\tup{x})}{}{\tup{\alpha}\tup{x}(\tup{\beta}(\tup{a}\tup{x})\tup{v})}\to \dt{P_\forall(\tup{a}\tup{x})}{}{\tup{\beta}(\tup{a}\tup{x})\tup{v}}\right)
\]
and this is immediate from the left hand premise of $(W_\prec)$ and the fact that $\phi(\tup{x})$ holds, and so we have completed the induction step and therefore the proof.
\end{proof}

\begin{proof}[Proof of Theorem \ref{res:operational}]
For $\loopskip$ and the primitive commands this is immediate. The core of the proof lies in the composition rule. Here, for any $s$, by the induction hypothesis there exist $s',\sigma',\Gamma'$ and $s'',\sigma',\Gamma'$ such that
\[
\begin{cases}
\forwards{s}{C_1}{s'}{\sigma}{\Gamma} \\
\forwards{s'}{C_2}{s''}{\sigma'}{\Gamma'}
\end{cases}
\]
and therefore
\[
\forwards{s}{\seq{C_1}{C_2}}{s''}{\sigma'::\sigma}{\Gamma'::\Gamma}
\]
for $s'=C_1^+s$ and $s''=C_2^+s'=(C_2^+\circ C^+_1)s=(\seq{C_1}{C_2})^+s$. For the backward direction, again using the induction hypothesis, for any $t,\sigma_0$ and $\Gamma_0$ we have
\[
\begin{cases}
\backward{\sigma'::\sigma::\sigma_0}{\Gamma'::\Gamma::\Gamma_0}{t}{\sigma::\sigma_0}{\Gamma::\Gamma_0}{t'}\\
\backward{\sigma::\sigma_0}{\Gamma::\Gamma_0}{t'}{\sigma_0}{\Gamma_0}{t''}
\end{cases}
\]
and therefore
\[
\backward{\sigma'::\sigma::\sigma_0}{\Gamma'::\Gamma::\Gamma_0}{t}{\sigma_0}{\Gamma_0}{t''}
\]
for $t'=C_2^-s't$ and $t''=C_1^-st'=C_1^-s(C_2^-s't)=C_1^-s(C_2^-(C_1^+s)t)=(\seq{C_1}{C_2})^-st$. The conditionals are straightforward, since whenever $\phi(s)$ then in $\weha$ we have
\begin{align*}
(\loopite{\phi}{C_1}{C_2})^+s&=C_1^+s \\
(\loopite{\phi}{C_1}{C_2})^-st&=C_1^-st
\end{align*}
and similarly for $\neg\phi(s)$. Finally, for the while loop we use induction on $\prec$, using that if $\phi(s)$ then $(\loopwhiledo{\prec}{\phi}{C})s=(\seq{C}{\loopwhiledo{\prec}{\phi}{C}})$ (and also $(\loopwhiledo{\prec}{\phi}{C})s=(\loopskip)s$ if $\neg\phi(s)$), and so the induction step is essentially the same as the composition rule.
\end{proof}

\end{document}